\documentclass[11pt,a4paper]{article}

 \usepackage{fullpage,authblk}

\title{Foundation for a Series of Efficient Simulation Algorithms}
\author{G\'erard C\'ec\'e}
\affil{\small{}
FEMTO-ST Insitute/CNRS,
Univ. Bourgogne Franche-Comt\'e, Montb\'eliard  France\\    
\texttt{Gerard.Cece@femto-st.fr}}

\newcommand{\Figure}{Figure~}
\newcommand{\Olics}[1]{}
\newcommand{\OarXiv}[1]{#1}

   \date{}

\usepackage{relsize}
\usepackage{exscale}

\usepackage{amsmath,amssymb,stmaryrd,latexsym,euscript,array,amsthm}

\usepackage{graphicx}

\usepackage[utf8]{inputenc}                                                

\usepackage{microtype}%if unwanted, comment out or use option "draft"

\Olics{\bibliographystyle{ieeetr}}% the recommended bibstyle
\OarXiv{\bibliographystyle{alpha}}

\usepackage{comment}

\usepackage{mathrsfs}

\usepackage{url}
\usepackage{hyperref} 

\usepackage[ruled,vlined,linesnumbered]{algorithm2e}
\usepackage{tikz}
\usetikzlibrary{calc}
\usetikzlibrary{patterns}
\usepgflibrary{shapes.misc} % for the shape "cross out"
\usepgflibrary{patterns}
\usepgflibrary{arrows} % pour les transitions maximales

\newcommand{\bline}{\vspace{1\baselineskip}} %rajoute une demi-ligne blanche

\newcommand{\suchthat}{\,\big|\,}
\usepackage{pifont}

\theoremstyle{plain}

\newtheorem{definition}{Definition}

\newtheorem{theorem}[definition]{Theorem}
\newtheorem{lemma}[definition]{Lemma}

\newtheorem{proposition}[definition]{Proposition}
\theoremstyle{remark}
\newtheorem*{remark}{\textbf{Remark}}
\DeclareMathOperator{\initRefine}{InitRefine}

\SetKwData{Rel}{Rel}
\SetKwData{Rep}{rep}
\SetKwData{SplitCount}{splitCount}
\SetKwData{NotRel}{NotRel}
\SetKwData{Block}{block}
\SetKwData{State}{state}
\SetKwData{Node}{node}
\SetKwData{True}{true}
\SetKwData{False}{false}
\SetKwData{Count}{count}
\SetKwData{RelCountAlgo}{RelCount}
\SetKwFunction{NewNode}{newNode()}
\SetKwFunction{NewBlock}{newBlock()}
\SetKwFunction{RelCount}{RelCount}
\SetKwFunction{Copy}{copy}
\SetKwFunction{ChooseState}{ChooseState}
\SetKwFunction{ChooseBlock}{ChooseBlock}
\SetKwFunction{Split}{Split}
\SetKwFunction{Swap}{swap}
\SetKwInput{Assert}{Assert}

\usepackage{environ}
\NewEnviron{killcontents}{}
\usepackage{mathptmx}
\usepackage{tikz}
\usepackage{mathrsfs}

\usetikzlibrary{angles} 
\usetikzlibrary{quotes}
\usetikzlibrary{positioning}
\usetikzlibrary{shapes.misc} % pour rounded rectangle and for the shape "cross out"
 \usetikzlibrary{calc}
  \usetikzlibrary{graphs}
  \usetikzlibrary{ arrows.meta }
  \usetikzlibrary{bending} 
 \usepgflibrary{shapes.geometric} % Pour semicircle
  \usetikzlibrary{math,fit}
   \usetikzlibrary{babel}
 \usetikzlibrary[arrows.meta]
 \usepgflibrary{shapes.multipart}

\DeclareSymbolFont{symbolsb}{OMS}{cmsy}{m}{n}
\SetSymbolFont{symbolsb}{bold}{OMS}{cmsy}{b}{n}
\DeclareSymbolFontAlphabet{\mathcal}{symbolsb}

\tikzset{
  maxtr/.style={near end,auto,swap,,circle,inner sep=1pt,minimum width=0,font=\tiny},
  dotstate/.style={inner sep=0,minimum width=2pt,fill,circle},
}
\tikzset{
  split/.initial=2pt, % droit etre defini en dehors du style qui suit
  up label/.initial=, % optional up label 
  down label/.initial=, % optional down label
  nosplit/.is if=nosplit,
 state end/.tip = {Circle[sep=1pt,length=0pt] . Circle[sep=-1pt,length=2pt]},
  pics/splitnode/.style={
    code={
      \pgfkeys{/tikz/split/.get=\distance, 
        }

\ifnosplit
    \xdef\distance{0}
    \tikzset{drawNode/.style={pic actions}, drawSemi/.style={draw=none,fill=none}}
\else
    \tikzset{drawNode/.style={draw=none,fill=none},        drawSemi/.style=pic actions}
\fi

\node[circle split,transparent] (-Origine) 
       {\pgfkeys{/tikz/up label}\nodepart{lower}\pgfkeys{/tikz/down label}} 
       let \p{w}=(-Origine.west), \p{e}=(-Origine.east),  \n{h}={(\x{e} - \x{w})} in 
           \pgfextra{\xdef\diam{\n{h}}};

\node [drawSemi,above=\distance, minimum width=\diam,anchor=south,semicircle]   (-Up)  
           {\pgfkeys{/tikz/up label}};
\node [drawSemi,above=-\distance, minimum width=\diam,anchor=north,semicircle, 
            shape border rotate=180] (-Down)
           {\pgfkeys{/tikz/down label}};   

     \node [drawNode,circle,anchor=center,minimum width=\diam,inner sep=0] (-Origine) {};
    }
 }
}

\begin{document}

\maketitle

%\category{CR-number}{subcategory}{third-level}
%\keywords{simulation, partition refinement algorithms, transition systems}

\begin{abstract}
  Compute the coarsest simulation preorder included in an initial preorder is
  used to reduce the resources needed to analyze a given transition system. This
  technique is applied on many models like Kripke structures, labeled graphs,
  labeled transition systems or even word and tree automata. Let
  $(Q,\rightarrow)$ be a given transition system and
  $\mathscr{R}_{\mathrm{init}}$ be an initial preorder over $Q$.  Until now,
  algorithms to compute $\mathscr{R}_{\mathrm{sim}}$, the coarsest simulation
  included in $\mathscr{R}_{\mathrm{init}}$, are either memory efficient or
  time efficient but not both. In this paper we propose the foundation for a
  series of efficient simulation algorithms with the introduction of the notion
  of maximal transitions and the notion of stability of a preorder with
  respect to a coarser one.  As an illustration we solve an open problem by
  providing the first algorithm with the best published time complexity,
  $O(|P_{\mathrm{sim}}|.|{\rightarrow}|)$, and a bit space complexity in
  $O(|P_{\mathrm{sim}}|^2.\log(|P_{\mathrm{sim}}|) +|Q|.\log(|Q|))$, with
  $P_{\mathrm{sim}}$ the partition induced by $\mathscr{R}_{\mathrm{sim}}$.
\end{abstract}

\section{Introduction}
\label{sec:introduction}

The simulation relation has been introduced by
Milner \cite{Milner71} as a behavioural relation between process.  This relation
can also be used to speed up the test of inclusion of languages \cite{ACHMV10}
or as a sufficient condition when this test of inclusion is undecidable in
general \cite{CG11}. Another very helpful use of a simulation relation is to
exhibit an equivalence relation over the states of a system. This allows to
reduce the state space of the given system to be analyzed while preserving an
important part of its properties, expressed in temporal logics for examples
\cite{GL94}. Note that the simulation equivalence yields a better reduction of
the state space than the better known bisimulation equivalence.

\subsection{State of the Art}
\label{sec:state-art}

The paper that has most influenced the literature is that of Henzinger,
Henzinger and Kopke \cite{HHK95}. Their algorithm, designed over Kripke
structures, and here named HHK, to compute $\mathscr{R}_{\mathrm{sim}}$, the
coarsest simulation, runs in $O(|Q|.|{\rightarrow}|)$-time, with $\rightarrow$
the transition relation over the state space $Q$, and uses $O(|Q|^2.\log(|Q|))$
bits.

But it happens that $\mathscr{R}_{\mathrm{sim}}$ is a preorder. And as such, it
can be more efficiently represented by a partition-relation pair $(P,R)$ with
$P$ a partition, of the state space $Q$, whose blocks are classes of the simulation equivalence relation
and with $R\subseteq P\times P$ a preorder over the blocks of $P$.  Bustan and
Grumberg \cite{BG03} used this to propose an algorithm, here named BG, with an
optimal bit-space complexity in $O(|P_{\mathrm{sim}}|^2 +
|Q|.\log(|P_{\mathrm{sim}}|))$ with $|P_{\mathrm{sim}}|$ (in general significantly
smaller than $|Q|$) the number of blocks of the partition $P_{\mathrm{sim}}$
associated with $\mathscr{R}_{\mathrm{sim}}$. Unfortunately, BG suffers from a
very bad time complexity. Then, Gentilini, Piazza and Policriti \cite{GPP03}
proposed an algorithm, here named GPP, with a better time complexity, in
$O(|P_{\mathrm{sim}}|^2.|{\rightarrow}|)$, and a claimed bit space complexity
like the one of BG.  This algorithm had a mistake and was corrected in
\cite{GP08}. It is very surprising that none of the authors citing \cite{GPP03},
including these of \cite{GP08,RT07,RT10,CRT11} and \cite{GPP15}, realized that the
announced bit space complexity was also not correct. Indeed, as shown in
\cite{Cec13a} and \cite{Ran14} the real bit space complexity of GPP is
$O(|P_{\mathrm{sim}}|^2.\log(|P_{\mathrm{sim}}|)+|Q|.\log(|Q|))$. In a similar way,
\cite{RT10} and \cite {CRT11} did a minor mistake by considering that a bit
space in $O(|Q|.\log(|P_{\mathrm{sim}}|))$ was sufficient to represent the
partition in their algorithms while a space in $O(|Q|.\log(|Q|))$ is needed.

Ranzato and Tapparo \cite{RT07,RT10} made a major breakthrough with their
algorithm, here named RT, which runs in
$O(|P_{\mathrm{sim}}|.|{\rightarrow}|)$-time but uses
$O(|P_{\mathrm{sim}}|.|Q|.\log(|Q|))$ bits, which is more than GPP. The
difficulty of the proofs and the abstract interpretation framework put aside, RT
is a reformulation of HHK but with a partition-relation pair instead of a mere
relation between states.  Over unlabelled transition systems, this is the best
algorithm regarding the time complexity.

\bline \emph{ Since \cite{RT07} a question has emerged: is there an algorithm
  with the time complexity of RT while preserving the space complexity of GPP ?
} \bline

Crafa, Ranzato and Tapparo \cite{CRT11}, modified RT to enhance its space
complexity. They proposed an algorithm with a time complexity in
$O(|P_{\mathrm{sim}}|.|\rightarrow|+|P_{\mathrm{sim}}|^2.|\rightarrow_{P_{\mathrm{sp}},P_{\mathrm{sim}}}|)$
and a bit space complexity in
$O(|P_{\mathrm{sp}}|.|P_{\mathrm{sim}}|.\log(|P_{\mathrm{sp}}|)+ |Q|.\log(|Q|))$
with $|P_{\mathrm{sp}}|$ between $|P_{\mathrm{sim}}|$ and $|P_{\mathrm{bis}}|$,
the number of bisimulation classes, and
$\rightarrow_{P_{\mathrm{sp}},P_{\mathrm{sim}}}$ a smaller abstraction of
$\rightarrow$. Unfortunately (although this algorithm provided new insights),
for 22 examples, out of the 24 they provided, there is no difference between
$|P_{\mathrm{bis}}|$, $|P_{\mathrm{sp}}|$ and $|P_{\mathrm{sim}}|$. For the two
remaining examples the difference is marginal.  With a little provocation, we
can then consider that $|P_{\mathrm{sp}}|\approx|P_{\mathrm{bis}}|$ and compute
the bisimulation equivalence (what should be done every time as it produces a
considerable speedup) then compute the simulation equivalence with GPP on the
obtained system is a better solution than the algorithm in \cite{CRT11} even if
an efficient computation of the bisimulation equivalence requires, see
\cite{PT87}, a bit space in $O(|\rightarrow|.\log(|Q|))$.

Ranzato \cite{Ran14} almost achieved the challenge by announcing an algorithm
with the space complexity of GPP but with the time complexity of RT multiplied
by a $\log(|Q|)$ factor. He concluded that the suppression of this $\log(|Q|)$
factor seemed to him quite hard to achieve. 
Gentilini, Piazza and Policriti \cite{GPP15} outperformed the
challenge by providing an algorithm with the space complexity of BG and 
the time complexity of RT,  but only in the special case of acyclic transition 
systems.

\subsection{Our Contributions}
\label{sec:our-contributions}
In this paper, we respond positively to the question and propose the
first simulation algorithm with the time complexity of RT and the space
complexity of GPP.

Our main sources of inspiration are \cite{PT87} for its implicit notion of
stability against a coarser partition, that we generalize in the present paper
for preorders, and for the counters it uses, \cite{HHK95} for the extension of
these counters for simulation algorithms,
\cite{BG03} for its use of little
brothers to which we prefer the use of what we define as maximal transitions,
\cite{Ran14} for its implicit use of maximal transitions to split blocks and
for keeping as preorders the intermediate relations of its algorithm and
\cite{Cec13a} for its equivalent definition of a simulation in terms of
compositions of relations.
 
 Note that almost all simulation algorithms are defined for Kripke
 structures. However, in each of them, after an initial step which consists in
 the construction of an initial preorder $\mathscr{R}_{\mathrm{init}}$, the
 algorithm is equivalent to calculating the coarsest simulation inside
 $\mathscr{R}_{\mathrm{init}}$ over a classical transition system.  We therefore directly start from a transition
 system $(Q,\rightarrow)$ and an initial preorder $\mathscr{R}_{\mathrm{init}}$
 inside which we compute the coarsest simulation.

 \Olics{
 \begin{remark}
   The proofs which are not in the main text may be found in \cite{Cec17a}.
%   \gnote{mettre la bonne reference (adresse web)}
 \end{remark}
}

%  \OarXiv{
%  \begin{remark}
%    The proofs which are not in the main text are in the appendix.
%  \end{remark}
% }

\section{Preliminaries}
\label{sec:preliminaries}

Let $Q$ be a set of elements, or \emph{states}. The number of elements of $Q$ is
denoted $|Q|$. A \emph{relation} over $Q$ is a subset of $Q\times Q$. Let
$\mathscr{R}$ be a relation over $Q$.  For all $q,q'\in Q$ we may write
$q\,\mathscr{R}\,q'$, or $q \mathbin{\tikz[baseline] \draw[dashed,->] (0pt,.5ex)
  -- node[font=\footnotesize,fill=white,inner sep=2pt] {$\mathscr{R}$}
  (6ex,.5ex);} q'$ in the figures, when $(q,q')\in\mathscr{R}$.  We define
$\mathscr{R}(q)\triangleq\{q'\in Q\suchthat q\,\mathscr{R}\,q'\}$ for $q\in Q$,
and $\mathscr{R}(X)\triangleq\cup_{q\in X}\mathscr{R}(q)$ for $X\subseteq Q$. We
write $X\mathrel\mathscr{R}Y$, or $X \mathbin{\tikz[baseline] \draw[dashed,->]
  (0pt,.7ex) -- node[font=\footnotesize,fill=white,inner sep=2pt]
  {$\mathscr{R}$} (6ex,.7ex);} Y$ in the figures, when $X\times
Y\cap\,\mathscr{R}\neq\emptyset$. For $q\in Q$ and $X\subseteq Q$, we also write
$X\mathrel\mathscr{R} q$ (resp.  $q\mathrel\mathscr{R} X$) for
$X\mathrel\mathscr{R} \{q\}$ (resp.  $\{q\}\mathrel\mathscr{R} X$).  A relation
$\mathscr{R}$ is said \emph{coarser} than another relation $\mathscr{R}'$ when
$\mathscr{R}'\subseteq\mathscr{R}$.  The \emph{inverse} of $\mathscr{R}$ is
${\mathscr{R}}^{-1}\triangleq\{(y,x)\in Q\times Q\suchthat (x,y)\in
\mathscr{R}\}$. The relation $\mathscr{R}$ is said \emph{symmetric} if
${\mathscr{R}}^{-1}\subseteq \mathscr{R}$ and \emph{antisymmetric} if
$q\,\mathscr{R}\,q'$ and $q'\,\mathscr{R}\,q$ implies $q=q'$.  Let $\mathscr{S}$
be a second relation over $Q$, the \emph{composition} of $\mathscr{R}$ by
$\mathscr{S}$ is $\mathscr{S}\mathrel{\circ}\mathscr{R}\triangleq\{(x,y)\in
Q\times Q\suchthat y\in\mathscr{S}(\mathscr{R}(x))\}$. The relation
$\mathscr{R}$ is said \emph{reflexive} if for all $q\in Q$ we have
$q\,\mathscr{R}\,q$, and \emph{transitive} if
$\mathscr{R}\mathrel{\circ}\mathscr{R}\subseteq\mathscr{R}$.  A \emph{preorder}
is a reflexive and transitive relation.  A \emph{partition} $P$ of $Q$ is a set
of non empty subsets of $Q$, called \emph{blocks}, that are pairwise disjoint
and whose union gives $Q$. A \emph{partition-relation pair} is a pair $(P,R)$
with $P$ a partition and $R$ a relation over $P$. To a partition-relation pair
$(P,R)$ we associate a relation $\mathscr{R}_{(P,R)}\triangleq \bigcup_{(C,D)\in
  R}C\times D$.  Let $\mathscr{R}$ be a preorder on $Q$ and $q\in Q$, we define
$[q]_\mathscr{R} \triangleq \{q'\in Q \suchthat q \,\mathscr{R}\, q' \wedge q'
\,\mathscr{R}\, q\}$ and $P_{\mathscr{R}} \triangleq \{[q]_\mathscr{R}\subseteq
Q\suchthat q\in Q\}$. It is easy to show that $P_{\mathscr{R}}$ is a partition
of $Q$. Therefore, given any preorder $\mathscr{R}$ and a state $q\in Q$, we
also call \emph{block}, the \emph{block of} $q$, the set $[q]_\mathscr{R}$.  A
symmetric preorder $\mathscr{P}$ is totally represented by the partition
$P_{\mathscr{P}}$ since $\mathscr{P}=\cup_{E\in P_\mathscr{P}} E\times
E$. Let us recall that a symmetric preorder is traditionally named an \emph{equivalence relation}.  Conversely, given a partition $P$, there is an associated equivalence relation
$\mathscr{P}_P\triangleq \cup_{E\in P}E\times E$. In the general case, a
preorder $\mathscr{R}$ is efficiently represented by the partition-relation pair
$(P_{\mathscr{R}}, R_{\mathscr{R}})$ with
$R_{\mathscr{R}}\triangleq\{([q]_\mathscr{R},[q']_\mathscr{R})\in
P_{\mathscr{R}}\times P_{\mathscr{R}}\suchthat q\,\mathscr{R}\,q'\}$ a
reflexive, transitive and antisymmetric relation over
$P_{\mathscr{R}}$. Furthermore, for a preorder $\mathscr{R}$, we note
$[\cdot]_{\mathscr{R}}$ the relation over $Q$ which associates to a state the
elements of its block. Said otherwise:
$[\cdot]_{\mathscr{R}}\triangleq\cup_{q\in
  Q}\{q\}\times[q]_{\mathscr{R}}$. Finally, for a set $X$ of sets we note $\cup
X$ for $\cup_{E\in X}E$.

\begin{proposition}
  \label{prop:RexistForAll}
  Let $X$ and $Y$ be two blocks of a preorder $\mathscr{R}$. Then
  \begin{displaymath}
    (X'\subseteq X\wedge Y'\subseteq Y\wedge X'\mathrel\mathscr{R}Y')\Rightarrow
    X\times Y\subseteq \mathscr{R}.
  \end{displaymath}
  Said otherwise, when two subsets of two blocks of $\mathscr{R}$ are related by
  $\mathscr{R}$ then all the elements of the first block are related by
  $\mathscr{R}$ with all the elements of the second block.  
\end{proposition}
\begin{proof}
  Thanks to the transitivity of  $\mathscr{R}$.
\end{proof}

A \emph{finite transition systems (TS)} is a pair $(Q,\rightarrow)$ with $Q$ a
finite set of states, and $\rightarrow$ a
relation over $Q$ called the \emph{transition relation}. A relation
$\mathscr{S}$ is a \emph{simulation} over 
$(Q,\rightarrow)$ if: 
\begin{equation}
  \label{eq:simulation}
  \mathscr{S}\mathrel\circ\rightarrow^{-1}\,\subseteq\,
  \rightarrow^{-1}\mathrel\circ \mathscr{S}
\end{equation}
For a simulation $\mathscr{S}$, when we have $q\mathrel\mathscr{S}q'$, we say
that $q$ is \emph{simulated} by $q'$ (or $q'$ \emph{simulates} $q$).

A relation $\mathscr{B}$ is a \emph{bisimulation} if $\mathscr{B}$ and
$\mathscr{B}^{-1}$ are both simulations. The interesting bisimulations, such as
the coarsest one included in a preorder, are equivalence relations. It is easy to
show that an equivalence relation $\mathscr{B}$ is a bisimulation iff :
\begin{equation}
 \mathscr{B}\mathrel\circ\rightarrow^{-1}\,\subseteq\,
  \rightarrow^{-1}\mathrel\circ \mathscr{B}
\end{equation}

\begin{remark}
  The classical definition is to say that a relation $\mathscr{S}$ is a
  simulation if: $q_1\mathrel\mathscr{S} q_2\wedge q_1\rightarrow q'_1
  \Rightarrow \exists q'_2\;.\;q_2\rightarrow q'_2\wedge
  {q'_1\mathrel\mathscr{S} q'_2} $. However, we prefer the formula
  \eqref{eq:simulation}, which is equivalent, because it is more global and to
  design efficient simulation algorithms we must abstract from individual states.
\end{remark}

In the remainder of the paper, all relations are over the same finite set $Q$
and the underlying transition system is $(Q,\rightarrow)$.

\section{Key Ideas}
\label{sec:ideas}

Let us start from \textbf{equation \eqref{eq:simulation}.} If a relation
$\mathscr{R}$ is not a simulation, we have
$\mathscr{R}\mathrel\circ{\rightarrow^{-1}}\nsubseteq\rightarrow^{-1}\mathrel{\circ}\mathscr{R}$. This
implies the existence of a relation $Remove$ such that:
$\mathscr{R}\mathrel\circ\,\rightarrow^{-1}\subseteq{(\rightarrow^{-1}\mathrel{\circ}\mathscr{R)}\cup
  Remove}$.  It can be shown that most of the simulation algorithms cited in the
introduction, like HHK, GPP and RT, are based on this last equation. In this
paper, like in \cite{Cec13a}, we make the following different choice. When
$\mathscr{R}$ is not a simulation, we reduce the problem of finding the coarsest
simulation inside $\mathscr{R}$ to the case where there is a relation $\mathit{NotRel}$
such that:
$\mathscr{R}\mathrel\circ\rightarrow^{-1}\,\subseteq\,\rightarrow^{-1}\mathrel\circ(\mathscr{R}\cup
\mathit{NotRel})$.  Let us note $\mathscr{U}\triangleq \mathscr{R}\cup \mathit{NotRel}$. We will
say that $\mathscr{R}$ is\textbf{ $\mathscr{U}$-stable} since we have:

\begin{equation}
  \label{eq:keyIdeaStable}
  \mathscr{R}\mathrel\circ\rightarrow^{-1}\,\subseteq\,
\rightarrow^{-1}\mathrel\circ \mathscr{U}
\end{equation}
Our definition of stability is new. However, it is implicit in the bisimulation
algorithm of \cite[p. 979]{PT87} where, with the notations from \cite{PT87}, a
partition $Q$ is said stable with every block of a coarser partition $X$. Within
 our formalism we can say the same thing with the formula $
\mathscr{P}_{Q}\mathrel\circ\rightarrow^{-1}\,\subseteq\,
\rightarrow^{-1}\mathrel\circ \mathscr{P}_{X}$.

\begin{figure}[!t]
  \centering
  \includegraphics[scale=0.7,page=1]{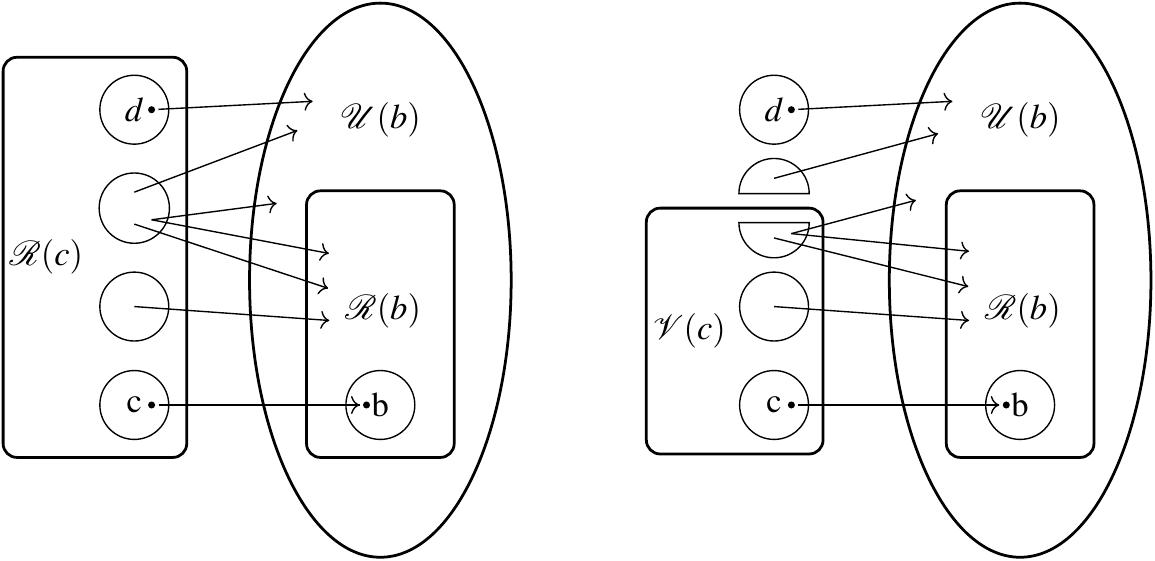}
  %\scalebox{.7}{\input{Ustable}}
 
  \caption{ $\mathscr{R}$ is $\mathscr{U}$-stable and $\mathscr{V}$, obtained
    after a split of blocks of $\mathscr{R}$  and a
    refinement of $\mathscr{R}$, is $\mathscr{R}$-stable.}
  \label{fig:intuiSim}
\end{figure}
 
Consider the transition $c\rightarrow b$ in \Figure\ref{fig:intuiSim}. The
preorder $\mathscr{R}$ is assumed to be $\mathscr{U}$-stable and we want to find
the coarsest simulation included in $\mathscr{R}$. Since $\mathscr{R}$ is a
preorder, the set $\mathscr{R}(c)$ is a union of blocks of $\mathscr{R}$. A
state $d$ in $\mathscr{R}(c)$ which doesn't have an outgoing transition to
$\mathscr{R}(b)$ belongs to $\rightarrow^{-1}\mathrel\circ\mathscr{U}(b)$,
thanks to \eqref{eq:keyIdeaStable}, but cannot simulate $c$. Thus, we can safely
remove it from $\mathscr{R}(c)$. But to do this effectively, we want to manage
blocks of states and not the individual states. Hence, we first do a
\textbf{split step} by splitting the blocks of $\mathscr{R}$ such that a
resulting block, included in both $\mathscr{R}(c)$ and
$\rightarrow^{-1}\mathrel\circ\mathscr{U}(b)$, is either completely included in
$\rightarrow^{-1}\mathrel\circ\mathscr{R}(b)$, which means that its elements
still have a chance to simulate $c$, or totally outside of it, which means that its
elements cannot simulate $c$. Let
us call $\mathscr{P}$ the equivalence relation associated to the resulting
partition. We will say that $\mathscr{P}$ is\textbf{
  $\mathscr{R}$-block-stable}. Then, to test whether a block, $E$ of
$\mathscr{P}$, which has an outgoing transition in
$\rightarrow^{-1}\mathrel\circ(\mathscr{U}\setminus\mathscr{R})(b)$, is included
in $\rightarrow^{-1}\mathrel\circ\mathscr{R}(b)$, it is sufficient to do the
test \textbf{for only one} of its elements, arbitrarily choosen, we call the
\textbf{representative} of $E$: $E.\Rep$. To do this test in constant time we
manage a counter which, at first, \textbf{count the number of transitions from
  $E.\Rep$ to $\mathscr{U}(b)=\mathscr{U}([b]_{\mathscr{P}})$}. By scanning the
transitions whose destination belongs to $(\mathscr{U}\setminus\mathscr{R})(b)$
this counter is updated to count the transitions from $E.\Rep$ to
$\mathscr{R}(b)=\mathscr{R}([b]_{\mathscr{P}})$. Therefore we get the
equivalences: there is no transition from $E$ to $\mathscr{R}(b)$ iff there is
no transition from $E.\Rep$ to $\mathscr{R}(b)$ iff this counter is null. Remark
that the total bit size of all the counters is in
$O(|P_{\mathrm{sim}}|^2.\log(|Q|))$ since there is at most $|P_{\mathrm{sim}}|$
blocks like $E$, $|P_{\mathrm{sim}}|$ blocks like $[b]_{\mathscr{P}}$ and $|Q|$
transitions from a state like $E.\Rep$. The difference is not 
so significative in practice but we will reduce this size to
$O(|P_{\mathrm{sim}}|^2.\log(|P_{\mathrm{sim}}|))$, at a cost of
$O(|P_{\mathrm{sim}}|.|\rightarrow|)$ elementary steps, which is hopefully
within our time budget.  Removing from $\mathscr{R}(c)$ the blocks of
$\rightarrow^{-1}\mathrel\circ(\mathscr{U}\setminus\mathscr{R})(b)$, like
$[d]_{\mathscr{P}}$, which do not have an outgoing transition to
$\mathscr{R}(b)$ is called the \textbf{refine step}. After this refine step,
$\mathscr{R}(c)$ has been reduced to $\mathscr{V}(c)$. Doing these split and
refine steps for all transitions $c\rightarrow b$ results in the relation
$\mathscr{V}$ that we will prove to be a $\mathscr{R}$-stable preorder.

In summary, from an initial preorder we will build a strictly decreasing series
of preorders $(\mathscr{R}_i)_{i\geq 0}$ such that $\mathscr{R}_{i+1}$ is
$\mathscr{R}_i$-stable and contains, by construction, all simulations included in
$\mathscr{R}_i$. Since all the relations are finite, this series has a limit,
reached in a finite number of steps. Let us call $\mathscr{R}_{\mathrm{sim}}$
this limit. We have: $\mathscr{R}_{\mathrm{sim}}$ is
$\mathscr{R}_{\mathrm{sim}}$-stable. Therefore, with \eqref{eq:simulation} and
\eqref{eq:keyIdeaStable}, $\mathscr{R}_{\mathrm{sim}}$ is a simulation and by
construction contains all simulations included in the initial preorder: this is
the coarsest one.

\begin{remark}
  The counters which are used in the previous paragraphs play a similar role as 
  the counters used in \cite{PT87}. Without them, the time complexity of the
  algorithm of the present paper would have been multiplied by a
  $|P_{\mathrm{sim}}|$ factor and  would have been this of GPP:
  $O(|P_{\mathrm{sim}}|^2.|{\rightarrow}|)$. 
\end{remark}

\section{Underlying Theory}
\label{sec:underlyingTh}

In this section we give the necessary theory to define what should be the ideal
split step and we justify the correctness of our refine step which allows to
treat blocks as if they were single states. We begin by introducing the notion
of maximal transition. This is the equivalent concept for transitions from that
of little brothers, introduced in \cite{BG03}, for states. The main difference
is that little brothers have been defined relatively to the final coarsest
simulation in a Kripke structure. Here we define maximal transitions
relatively to a current preorder $\mathscr{R}$.  

\begin{definition}
  \label{def:maxTrans}
  Let $\mathscr{R}$ be a preorder. The transition $q\rightarrow q'$ is said
  \emph{maximal for}  $\mathscr{R}$, or
  \emph{$\mathscr{R}$-maximal}, which is noted $q\rightarrow_{\mathscr{R}}q'$,
  when:
  \begin{displaymath}
  \forall q''\in Q \;.\;
  (q\rightarrow q'' \wedge q'\mathrel\mathscr{R}q'')\Rightarrow
  q''\in[q']_{\mathscr{R}}
\end{displaymath}
  The set of $\mathscr{R}$-maximal transitions and the induced relation are both noted
  $\rightarrow_{\mathscr{R}}$.
\end{definition}
\begin{figure}[!t]
  \centering    
  \includegraphics[scale=1]{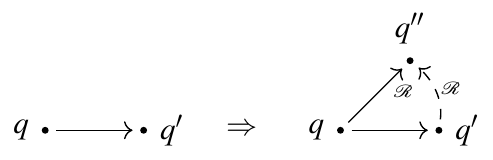}

  \caption{Illustration of the left property of Lemma~\ref{lem:transInMaxoR}.}
  \label{fig:transInMaxoR}
\end{figure}
\begin{lemma}[\Figure~\ref{fig:transInMaxoR}]
  \label{lem:transInMaxoR} 
  For a preorder $\mathscr{R}$, the two
  following properties are verified:  
  \begin{displaymath}
   \rightarrow^{-1}\,\subseteq\,
   \rightarrow_{\mathscr{R}}^{-1}\mathrel{\circ}\mathscr{R} \text{ and }
       \rightarrow^{-1}\mathrel{\circ}\mathscr{R} \,=\, \rightarrow^{-1}_{\mathscr{R}}\mathrel{\circ}\mathscr{R}
\end{displaymath} 
\end{lemma}

\begin{proof}%[Proof of Lemma~\ref{lem:transInMaxoR}]
  Let $(q,q')\in\rightarrow$ and $X=\{q''\in Q\suchthat q\rightarrow q''\wedge
  q'\mathscr{R}q''\}$.  Since $\mathscr{R}$ is reflexive, this set is not empty
  because it contains $q'$. Let $Y$ be the set of blocks of $\mathscr{R}$ which
  contain an element from $X$. Since this set is finite (there is a finite
  number of blocks) there is at least a block $G$ maximal in $Y$. Said
  otherwise, there is no $G'\in Y$, different from $G$, such that $G
  \,\mathscr{R}\, G'$. Let $q''\in G$ such that $q\rightarrow q''$. From what
  precedes, the transition $(q,q'')$ is maximal and $q'\mathscr{R}q''$. Hence:
  $(q',q)\in\rightarrow_{\mathscr{R}}^{-1}\mathrel{\circ}\mathscr{R}$. So we have:
  \begin{equation}
    \label{eq:transInMaxoR}
    \rightarrow^{-1} \,\subseteq\, \rightarrow^{-1}_{\mathscr{R}}\mathrel{\circ}\mathscr{R}
  \end{equation}

Now, from \eqref{eq:transInMaxoR} we get $
  \rightarrow^{-1}\mathrel{\circ}\mathscr{R} \subseteq
  \rightarrow^{-1}_{\mathscr{R}}\mathrel{\circ}\mathscr{R}\mathrel{\circ}\mathscr{R}$
  and thus  $
  \rightarrow^{-1}\mathrel{\circ}\mathscr{R} \subseteq
  \rightarrow^{-1}_{\mathscr{R}}\mathrel{\circ}\mathscr{R}$ since $\mathscr{R}$
  is a preorder. The relation $\rightarrow_{\mathscr{R}}$ is a subset of
  $\rightarrow$. Therefore we also have
  $\rightarrow^{-1}_{\mathscr{R}}\mathrel{\circ}\mathscr{R}\subseteq
  \rightarrow^{-1}\mathrel{\circ}\mathscr{R}$ which concludes the proof. 
  
\end{proof}

In the last section, we introduced the notions of stability and of
block-stability. Let us define them formaly.
%as $\mathscr{R}$-block-stable an
% equivalence relation  $\mathscr{P}\subseteq\mathscr{R}$ that satisfies: 
% \begin{multline}
%   \label{eq:sameBlock}
%   \forall b\in Q\,.\, (d,d')\in \mathscr{P}\Rightarrow \Olics{\\} (d\in  \rightarrow^{-1}\mathrel{\circ}\mathscr{R}(b) \Leftrightarrow
%   d'\in  \rightarrow^{-1}\mathrel{\circ}\mathscr{R}(b))
% \end{multline}

\begin{definition}
  \label{def:Stability}
  
  % A preorder $\mathscr{R}$ is said $\mathscr{U}$-\emph{stable}, with
  % $\mathscr{U}$ a coarser preorder than $\mathscr{R}$, if:
  %   \begin{equation}
  %     \label{def:relStability}
  %   \mathscr{R}\mathrel\circ\rightarrow^{-1}\,\subseteq\,
  %     \rightarrow^{-1}\mathrel\circ \mathscr{U}
  %   \end{equation}
  Let $\mathscr{R}$ a preorder.
 
  \begin{itemize}

    % \begin{equation}
    %   \label{def:partStability}
    %   [\cdot]_{\mathscr{P}}\mathrel\circ\rightarrow_{\mathscr{R}}^{-1}\,\subseteq\,
    %   \rightarrow_{\mathscr{R}}^{-1}\mathrel\circ  [\cdot]_{\mathscr{R}}
    % \end{equation}
  \item $\mathscr{R}$ is said
    $\mathscr{U}$-\emph{stable}, with $\mathscr{U}$ a coarser preorder than
    $\mathscr{R}$,  if:
    \begin{equation}
      \label{def:relStability}
    \mathscr{R}\mathrel\circ\rightarrow^{-1}\,\subseteq\,
      \rightarrow^{-1}\mathrel\circ \mathscr{U}
    \end{equation}
    
  \item An equivalence relation $\mathscr{P}$ included in $\mathscr{R}$, is said
    \emph{$\mathscr{R}$-block-stable} if:
    \begin{multline}
      \label{eq:sameBlock}
      \OarXiv{\hfill}  \forall b,d,d'\in Q\,.\, d\,\mathscr{P}\,d'\Rightarrow \Olics{\\} (d\in  \rightarrow^{-1}\mathrel{\circ}\mathscr{R}(b) \Leftrightarrow
      d'\in  \rightarrow^{-1}\mathrel{\circ}\mathscr{R}(b)) \OarXiv{\hfill} 
    \end{multline}
  \end{itemize}
\end{definition}

\begin{remark} Say that $\mathscr{P}$ is included in $\mathscr{R}$ means that
  each block of $\mathscr{P}$ is included in a block of $\mathscr{R}$.
  
  % \gnote{ Mettre plus tard :
  %   The formula \eqref{eq:partStability} is equivalent
  %   with: $e\rightarrow_{\mathscr{R}}e'\;\Rightarrow\;
  %   [e]_{\mathscr{P}}\subseteq
  %   \rightarrow^{-1}_{\mathscr{R}}([e']_{\mathscr{R}})$.}
\end{remark}

As seen in the following lemma we have a nice equivalence: an equivalence
relation $\mathscr{P}$ is $\mathscr{R}$-block-stable iff it is
$\mathscr{R}$-stable.

\begin{lemma}
  \label{lem:partStability}
  Let $\mathscr{P}$ be an equivalence relation included in a preorder
  $\mathscr{R}$. Then \eqref{eq:sameBlock} is equivalent with: 
  \begin{equation}
    \label{eq:partStability}
    \mathscr{P}\mathrel\circ\rightarrow^{-1}\,\subseteq\,
    \rightarrow^{-1}\mathrel\circ  \mathscr{R}
  \end{equation}  
\end{lemma}

\begin{proof}%[Proof of Lemma~\ref{lem:partStability}]
  \begin{figure}[!t]
    \centering    
    \includegraphics[scale=1]{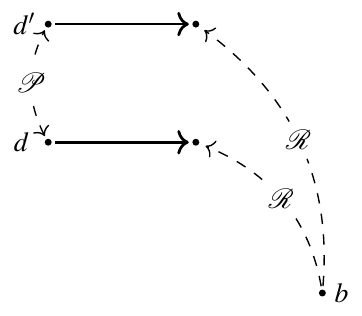}
    \caption{Illustration of \ref{lem:transInMaxoR}.}
    \label{fig:partStability}
  \end{figure}
  To show the equivalence of \eqref{eq:sameBlock} and \eqref{eq:partStability} we
  use an intermediate property:
  \begin{equation}
    \label{eq:intermediateSameBlock}
    \mathscr{P}\mathrel\circ\rightarrow^{-1}\mathrel\circ\mathscr{R}\,\subseteq\,
    \rightarrow^{-1}\mathrel\circ  \mathscr{R}
  \end{equation}  
  With the help of \Figure~\ref{fig:partStability} it is straightforward to see
  the equivalence of \eqref{eq:sameBlock} and \eqref{eq:intermediateSameBlock}.
  It remains therefore to show the equivalence of \eqref{eq:partStability} and
  \eqref{eq:intermediateSameBlock}.
  
  \begin{itemize}
  \item[] \eqref{eq:partStability} $\Rightarrow$
    \eqref{eq:intermediateSameBlock}. From \eqref{eq:partStability} we get
    $\mathscr{P}\mathrel\circ\rightarrow^{-1}\mathrel\circ \mathscr{R}\,\subseteq\,
    \rightarrow^{-1}\mathrel\circ \mathscr{R}\mathrel\circ \mathscr{R}$ and thus
    \eqref{eq:intermediateSameBlock} since, as a preorder $\mathscr{R}$ is
    transitive.         
    
  \item[] \eqref{eq:intermediateSameBlock} $\Rightarrow$
    \eqref{eq:partStability}.  Let $\mathscr{I}=\{(q,q)\suchthat q\in Q\}$ be
    the identity relation. We have $\mathscr{P}\mathrel\circ\rightarrow^{-1}=
    \mathscr{P}\mathrel\circ\rightarrow^{-1}\mathrel\circ\mathscr{I}$ and thus $
    \mathscr{P}\mathrel\circ\rightarrow^{-1}\mathrel\circ\mathscr{I}\,\subseteq\,
    \mathscr{P}\mathrel\circ\rightarrow^{-1}\mathrel\circ \mathscr{R}$ since
    $\mathscr{R}$ is a preorder and as such contains $\mathscr{I}$. With
    \eqref{eq:intermediateSameBlock} we thus get \eqref{eq:partStability}.        
  \end{itemize} 
\end{proof}

With \eqref{def:relStability} and \eqref{eq:partStability} the reader
should now be convinced by the interest of \eqref{eq:simulation} to define a
simulation.

Following the keys ideas given in Section~\ref{sec:ideas} there is an interest,
for the time complexity, of having a coarse $\mathscr{R}$-block-stable
equivalence relation $\mathscr{P}$. Hopefully there is a coarsest one.

\begin{proposition}
  Given a preorder $\mathscr{R}$, there is a coarsest
  $\mathscr{R}$-stable equivalence relation.
\end{proposition}
\begin{proof}
  With Lemma~\ref{lem:partStability} and by an easy induction based on the two
  following properties:  
  \begin{itemize}
  \item the identity relation, $\mathscr{I}=\{(q,q)\suchthat q\in Q\}$,  is a
    $\mathscr{R}$-stable equivalence relation.
  \item the reflexive and transitive closure
    $(\mathscr{P}_1\cup\mathscr{P}_2)^*$ of the union of two
    $\mathscr{R}$-stable equivalence relations, $\mathscr{P}_1$ and
    $\mathscr{P}_2$, is also a $\mathscr{R}$-stable equivalence relation,
    coarser than them.
  \end{itemize}
\end{proof}

We are now ready to introduce the main result of this section.  It is a
formalization, and a justification, of the refine step given in
Section~\ref{sec:ideas}. In the following theorem, the link with the decreasing
sequence of relations $(\mathscr{R}_i)_{i\geq 0}$ mentioned at the end of
Section~\ref{sec:ideas} is: if $\mathscr{R}_i$ is the current value of
$\mathscr{R}$ then $\mathscr{R}_{i-1}$ is  $\mathscr{U}$ and
$\mathscr{R}_{i+1}$ will be $\mathscr{V}$. The reader can also ease its
comprehension of the theorem by considering \Figure~\ref{fig:intuiSim}.

\begin{theorem}
  \label{th:refinement}  
  Let $\mathscr{U}$ be a preorder, $\mathscr{R}$ be a $\mathscr{U}$-stable
  preorder and $\mathscr{P}$ be the coarsest $\mathscr{R}$-stable equivalence
  relation. Let 
  $\mathit{NotRel}=\mathscr{U}\mathrel\setminus\mathscr{R}$ and
  $\mathscr{V}=\mathscr{R}\mathrel\setminus \mathit{NotRel}'$ with
  \begin{displaymath}
    \mathit{NotRel}' = 
    \bigcup_{\substack{b,c,d\,\in\, Q,\; c\,\rightarrow\, b,\;
        c\,\mathscr{R}\, d,\\
        d\,\in\, \rightarrow^{-1}\circ\, \mathit{NotRel}(b),\;\\ 
        d\,\not\in\, \rightarrow^{-1}\,\circ\,\mathscr{R}(b)}}
    [c]_{\mathscr{P}}{\times} [d]_{\mathscr{P}}
  \end{displaymath}
  Then:
  \begin{enumerate}
  \item\label{th:refinement:item:1}  $\mathit{NotRel}' = X$ with
    \begin{displaymath}
    X = 
    \bigcup_{\substack{b,c,d\,\in\, Q,\; c\,\rightarrow_{\mbox{\tiny $\mathscr{R}$}}\, b,\;
        c\,\mathscr{R}\, d,\\
        d\,\in\, {\rightarrow}_{ \mbox{\tiny $\mathscr{R}$}}^{-1}\circ\, \mathit{NotRel}(b),\;\\ 
        d\,\not\in\, \rightarrow_{\mbox{\tiny $\mathscr{R}$}}^{-1}\,\circ\,\mathscr{R}(b)}}
    \{(c,d)\}
\end{displaymath}

  \item Any simulation $\mathscr{S}$ included in $\mathscr{R}$ is also included
    in $\mathscr{V}$.
  \item\label{th:refinement:item:3}  $\mathscr{V}\mathrel\circ\rightarrow^{-1}\,\subseteq\,
    \rightarrow^{-1}\mathrel\circ \mathscr{R}$
    
  \item $\mathscr{V}$ is a preorder.    
  \item\label{th:refinement:item:5} $\mathscr{V}$ is $\mathscr{R}$-stable.

  \item\label{th:refinement:item:6}  Blocks of $\mathscr{V}$ are blocks of  $\mathscr{P}$ (i.e. $P_{\mathscr{V}}
    =  P_{\mathscr{P}}$).
  \end{enumerate}
\end{theorem}

\begin{proof}%[Proof of Theorem~\ref{th:refinement}]
  \mbox{}
  \begin{enumerate}
  \item Since $(c,d)$ belongs to $[c]_{\mathscr{P}}{\times} [d]_{\mathscr{P}}$,
    $\rightarrow_{\mathscr{R}} \;\subseteq \; \rightarrow$ and
    $\rightarrow^{-1}\mathrel{\circ}\mathscr{R} \,=\,
    \rightarrow^{-1}_{\mathscr{R}}\mathrel{\circ}\mathscr{R}$, from
    Lemma~\ref{lem:transInMaxoR}, we get $X\subseteq \mathit{NotRel}'$. For the converse,
    let $(c',d')\in \mathit{NotRel}'$. By definition, there are $b,c,d\in Q$ such that
    $c\rightarrow b$, $c\,\mathscr{R}\, d$, $d\in \rightarrow^{-1}\circ\,
    \mathit{NotRel}(b)$, $d\not\in\rightarrow^{-1}\mathrel\circ\mathscr{R}(b)$,
    $c'{\in}[c]_{\mathscr{P}}$ and $d'{\in} [d]_{\mathscr{P}}$. From
    $d\not\in\rightarrow^{-1}\mathrel\circ\mathscr{R}(b)$ and
    Lemma~\ref{lem:transInMaxoR} we have
    $d\not\in\rightarrow_{\mathscr{R}}^{-1}\mathrel\circ\mathscr{R}(b)$. From
    $c\rightarrow b$, Lemma~\ref{lem:transInMaxoR},
    Lemma~\ref{lem:partStability}, and the hypothesis that $\mathscr{P}$ is
    $\mathscr{R}$-stable, we have
    $c'\rightarrow_{\mathscr{R}}^{-1}\mathrel\circ\mathscr{R}(b)$. Therefore,
    there is a state $b'$ such that $c'\rightarrow_{\mathscr{R}}b'$ and
    $b\,\mathscr{R}\, b'$. Les us suppose
    $d'\in\rightarrow_{\mathscr{R}}^{-1}\mathrel\circ\mathscr{R}(b')$. Since
    $b\,\mathscr{R}\, b'$, we would have had
    $d\in\rightarrow_{\mathscr{R}}^{-1}\mathrel\circ\mathscr{R}(b)$. Thus
    $d'\not\in\rightarrow_{\mathscr{R}}^{-1}\mathrel\circ\mathscr{R}(b')$.  We
    have $c'\,\mathscr{P}\, c$, $c\,\mathscr{R}\, d$, $d\,\mathscr{P}\, d'$, and
    thus $c'\,\mathscr{R}\, d'$ since $\mathscr{R}$ is a preorder and
    $\mathscr{P}\subseteq\mathscr{R}$. With $c'\rightarrow_{\mathscr{R}}b'$ and
    the hypothesis that $\mathscr{R}$ is $\mathscr{U}$-stable, we get $d'\in
    \rightarrow^{-1}\mathrel\circ\mathscr{U}(b')$ and thus, with
    Lemma~\ref{lem:transInMaxoR}, $d'\in
    \rightarrow_{\mathscr{R}}^{-1}\mathrel\circ\mathscr{R}\mathrel\circ\mathscr{U}(b')$
    and thus $d'\in \rightarrow_{\mathscr{R}}^{-1}\mathrel\circ\mathscr{U}(b')$
    since $\mathscr{U}$ is a preorder and $\mathscr{R}\subseteq\mathscr{U}$. As
    seen above,
    $d'\not\in\rightarrow_{\mathscr{R}}^{-1}\mathrel\circ\mathscr{R}(b')$. So we
    have
    $d'\in\rightarrow_{\mathscr{R}}^{-1}\mathrel\circ(\mathscr{U}\setminus\mathscr{R})(b')$. In
    summary: $c'\rightarrow_{\mathscr{R}}b'$, $c'\,\mathscr{R}\, d'$,
    $d'\in\rightarrow_{\mathscr{R}}^{-1}\mathrel\circ \mathit{NotRel}(b')$ and
    $d'\not\in\rightarrow_{\mathscr{R}}^{-1}\mathrel\circ\mathscr{R}(b')$. All of this implies
    that $(c',d')\in X$. So we have $\mathit{NotRel}'\subseteq X$ and thus $\mathit{NotRel}'=X$.

  \item By contradiction.  Let $(c,d)\in \mathscr{S}$ such that $(c,d)\not\in
    \mathscr{V}$. This means that $(c,d)\in \mathit{NotRel}'$.  From
    \ref{th:refinement:item:1}) and the hypothesis
    $\mathscr{S}\subseteq\mathscr{R}$ there is $b\in Q$ such that
    $c\rightarrow_{\mathscr{R}} b$,
    $d\not\in\rightarrow_{\mathscr{R}}^{-1}\mathrel\circ\mathscr{R}(b)$ and thus
    $d\not\in\rightarrow^{-1}\mathrel\circ\mathscr{R}(b)$, from
    Lemma~\ref{lem:transInMaxoR}. From $c\rightarrow_{\mathscr{R}} b$, thus
    $c\rightarrow b$, and the assumption that $\mathscr{S}$ is a simulation
    there is $d'\in Q$ with $d\rightarrow d'$ and $b\mathrel\mathscr{S}d'$ thus
    $b\mathrel\mathscr{R}d'$. This contradicts
    $d\not\in\rightarrow^{-1}\mathrel\circ\mathscr{R}(b)$. Therefore
    $\mathscr{S}\subseteq \mathscr{V}$.    
  
  \item If this is not the case, there are $b,c,d\in Q$ such that
    $c\mathrel\mathscr{V}d$, $c\rightarrow b$ and
    $d\not\in\rightarrow^{-1}\mathrel\circ\mathscr{R}(b)$.
    Since $\mathscr{V}\subseteq\mathscr{R}$ and $\mathscr{R}$ is a
    $\mathscr{U}$-stable relation there is $d'\in Q$ such that $d\rightarrow d'$
    and $b\mathrel\mathscr{U}d'$. The case $b\mathrel\mathscr{R}d'$ would contradict
    $d\not\in\rightarrow^{-1}\mathrel\circ\mathscr{R}(b)$. Therefore
    $d\in\rightarrow^{-1}\mathrel\circ \mathit{NotRel}(b)$ and all the conditions are
    met for $(c,d)$ belonging in $\mathit{NotRel}'$ which contradicts
    $c\mathrel\mathscr{V}d$.
    
  \item Let us show that $\mathscr{V}$ is both reflexive and transitive. If it
    is not reflexive, since $\mathscr{R}$ is reflexive, from
    \ref{th:refinement:item:1}) there is $(c,d)$ in $X$ and a state $b$ such
    that $c\rightarrow_{\mathscr{R}} b$ and
    $d\not\in\rightarrow_{\mathscr{R}}^{-1}\mathrel\circ\mathscr{R}(b)$ and
    $c=d$. But this is impossible since $\mathscr{R}$ is reflexive. Hence,
    $\mathscr{V}$ is reflexive.  We also prove by contradiction that
    $\mathscr{V}$ is transitive. If it is not the case, there are $c,e,d\in Q$
    such that $c\mathrel\mathscr{V}e$, $e\mathrel\mathscr{V}d$ but $\mathrel\neg
    c\mathrel\mathscr{V}d$. Since $\mathscr{V}\subseteq\mathscr{R}$ and
    $\mathscr{R}$ is transitive then $c\mathrel\mathscr{R}d$. With $\mathrel\neg
    c\mathrel\mathscr{V}d$ and \ref{th:refinement:item:1}), there is $b$ such
    that $c\rightarrow_{\mathscr{R}} b$ and
    $d\not\in\rightarrow_{\mathscr{R}}^{-1}\mathrel\circ\mathscr{R}(b)$. But
    from \ref{th:refinement:item:3}), there is $b'\in Q$ such that
    $b\mathrel\mathscr{R}b'$ and $e\rightarrow b'$. With $e\mathrel\mathscr{V}d$
    and the same reason, there is $b''$ such that $b'\mathrel\mathscr{R}b''$ and
    $d\rightarrow b''$. By transitivity of $\mathscr{R}$ we get
    $b\mathrel\mathscr{R}b''$ and thus
    $d\in\rightarrow^{-1}\mathrel\circ\mathscr{R}(b)$. With
    Lemma~\ref{lem:transInMaxoR} this  contradicts
    $d\not\in\rightarrow_{\mathscr{R}}^{-1}\mathrel\circ\mathscr{R}(b)$. Hence, $\mathscr{V}$
    is transitive.
    
  \item This is a direct consequence of the two preceding items and the fact
    that by construction $\mathscr{V}\subseteq\mathscr{R}$.
    
  \item By hypothesis, $\mathscr{P}\subseteq\mathscr{R}$. This means that blocks
    of $\mathscr{R}$ are made of blocks of $\mathscr{P}$. By definition,
    $\mathscr{V}$ is obtained by deleting from $\mathscr{R}$ relations between
    blocks of $\mathscr{P}$. This implies that blocks of $\mathscr{V}$ are made
    of blocks of $\mathscr{P}$. To proove that a block of $\mathscr{V}$ is made
    of a single block of $\mathscr{P}$, let us assume, by contradiction, that
    there are two different blocks, $B_1$ and $B_2$, of $\mathscr{P}$ in a block
    of $\mathscr{V}$. We show that $\mathscr{P}$ is not the coarsest
    $\mathscr{R}$-block-stable equivalence relation.  Let
    $\mathscr{P}'=\mathscr{P}\cup B_1\times B_2\cup B_2\times B_1$. Then
    $\mathscr{P}'$ is an equivalence relation strictly coarser than
    $\mathscr{P}'$. Furthermore, since $B_1$ and $B_2$ are blocks of
    $\mathscr{V}$, we get $\mathscr{P}'\subseteq \mathscr{V}$. With
    \ref{th:refinement:item:3}) we get that $\mathscr{P}'$ is
    $\mathscr{R}$-stable and thus $\mathscr{R}$-block-stable with
    Lemma~\ref{lem:partStability}.  This contradicts the hypothesis that
    $\mathscr{P}$ was the coarsest one. Therefore, blocks of $\mathscr{V}$ are
    blocks of $\mathscr{P}$.
        
  \end{enumerate}
\end{proof}

\begin{remark}
  \ref{th:refinement:item:1}) means that blocks of $\mathscr{P}$ are
  sufficiently small to do the refinement step efficiently, as if they were
  states.  \ref{th:refinement:item:6}) means that these blocks cannot be bigger.
  \ref{th:refinement:item:5}) means that we are ready for next split and
  refinement steps.
\end{remark}

In what precedes, we have assumed that for the preorder $\mathscr{R}$, inside
which we want to compute the coarsest simulation, there is another preorder
$\mathscr{U}$ such that condition \eqref{def:relStability} holds. The fifth item
of Theorem~\ref{th:refinement} says that if this true at a given iteration (made
of a split step and a refinement step) of the algorithm then this is true at the
next iteration. For the end of this section we show that we can safely modify
the initial preorder such that this is also initially true.  This is indeed a
simple consequence of the fact that a state with an outgoing transition cannot
be simulated by a state with no outgoing transition.

\begin{definition}
  \label{def:InitRefine}
  Let $\mathscr{R}_{\mathrm{init}}$ be a preorder. We define
  $\initRefine(\mathscr{R}_{\mathrm{init}})$ such that:
  \begin{multline*}
    \OarXiv{\hfill} \initRefine(\mathscr{R}_{\mathrm{init}})\triangleq 
   \mathscr{R}_{\mathrm{init}}\cap
    \{(c,d)\in Q\times Q\suchthat \Olics{\\} \exists c'\in Q\,.\,c\rightarrow c' \;\Rightarrow\;
        \exists d'\in Q\,.\,d\rightarrow d'    \} \OarXiv{\hfill} 
  \end{multline*}
\end{definition}

\begin{proposition}
  \label{prop:InitRefine}
  Let  $\mathscr{R} =
  \initRefine(\mathscr{R}_{\mathrm{init}})$ with $\mathscr{R}_{\mathrm{init}}$ a preorder. Then:
  \begin{enumerate}
  \item $\mathscr{R}$ is $(Q\times Q)$-stable,
  \item a simulation $\mathscr{S}$ included in $\mathscr{R}_{\mathrm{init}}$ is also
    included in $\mathscr{R}$.
  \end{enumerate}
\end{proposition}

\begin{proof}%[Proof of Proposition~\ref{prop:InitRefine}]
  \mbox{}
  \begin{enumerate}
  \item
    $Q\times Q$ is trivially a preorder. It remains to show that $\mathscr{R}$ is
    also a preorder and that \eqref{def:relStability} is true with
    $\mathscr{U}=Q\times Q$.
    
    Since $\mathscr{R}_{\mathrm{init}}$ is a preorder and thus reflexive,
    $\mathscr{R}$ is also trivially reflexive. Now, by contradiction, let us
    suppose that $\mathscr{R}$ is not transitive. There are three states $c,e,d\in Q$
    such that: $c\,\mathscr{R}\,e\wedge e\,\mathscr{R}\,d \wedge \neg\;
    c\,\mathscr{R}\,d$. From the fact that
    $\mathscr{R}\subseteq\mathscr{R}_{\mathrm{init}}$ and $\mathscr{R}_{\mathrm{init}}$ is a
    preorder, we get $c\,\mathscr{R}_{\mathrm{init}}\,d$. With $\neg\; c\,\mathscr{R}\,d$
    and the definition of $\mathscr{R}$ this means that $c$ has a successor
    while $d$ has not. But the hypothesis that $c$ has a successor and
    $c\,\mathscr{R}\,e$ implies that $e$ has a successor. With
    $e\,\mathscr{R}\,d$ we also get that $d$ has also a successor, which contradicts
    what is written above. Hence, $\mathscr{R}$ is transitive and thus a
    preorder.

    The formula $\mathscr{R}\mathrel{\circ}\rightarrow^{-1}\subseteq
    \rightarrow^{-1}\circ (Q\times Q)$ just means that the two hypotheses,
    $c\,\mathscr{R}\,d$ and $c$ has a successor, imply that $d$ has also a
    successor. This is exactly the meaning of the second part of the
    intersection in the definition of $\mathscr{R}$.    
     
  \item By contradiction, if this is not true there is $(c,d)$ a pair of states
    which belongs to $\mathscr{S}$ and $\mathscr{R}_{\mathrm{init}}$ but does not belong
    to $\mathscr{R}$. By definition of $\mathscr{R}$ this means that $c$ has a
    successor while $d$ has not. But the hypotheses $c\,\mathscr{S}\,d$, $c$ has
    a successor and $\mathscr{S}$ is a simulation imply that $d$ has also a
    successor which contradicts what is written above. Hence, $\mathscr{S}$ is
    also included in $\mathscr{R}$.    
   \end{enumerate}
\end{proof}

The total relation $Q\times Q$  will thus play the role of the
initial $\mathscr{U}$ in the algorithm.

\begin{remark}
  In \cite[p. 979]{PT87} there is also a similar preprocessing of the initial
  partition where states with no output transition are put aside.
\end{remark}

\section{The Algorithm}
\label{sec:algo}

The approach of the previous section can be applied to several algorithms, with
different balances between time complexity and space complexity.  It can also be
extended to labelled transition systems.  In this section the emphasis is on
the, theoretically, most efficient in memory of the fastest simulation
algorithms of the moment.

\subsection{Counters and Splitter Transitions}
\label{sec:splitter-transitions}

Let us remember that a partition $P$ and its associated equivalence relation
$\mathscr{P}_P$ (or a equivalence relation $\mathscr{P}$ and its associated
partition $P_{\mathscr{P}}$) denote essentially the same thing. The difference
is that for a partition we focus on the set of blocks whereas for an equivalence
relation we focus on the relation which relates the elements of a same block.
For a first reading, the reader may consider that a partition is an equivalence relation, and vice versa.

From a preorder $\mathscr{R}$ that satisfies \eqref{def:relStability} we will
need to split its blocks in order to find its coarsest
$\mathscr{R}$-block-stable equivalence relation. Then, Theorem~\ref{th:refinement}
will be used for a refine step. For all this, we first need the traditional
\Split function.

\begin{definition}
   \label{def:split}   
   Given a partition $P$ and a set of states $Marked$, the function
   $\Split(P,Marked)$ returns a partition similar to $P$, but with the
   difference that each block $E$ of $P$ such that $E\cap Marked \neq\emptyset$
   and $E\not\subseteq Marked$ is replaced by two blocks: $E_1=E\cap Marked$ and
   $E_2=E\mathrel\setminus E_1$.   
\end{definition}

To efficiently perform the split and refine steps we need a set of counters
which associates to each representative state of a block, of an equivalence relation,
the number of blocks, of that same equivalence relation, it
reaches in $\mathscr{R}(B)$, for $B$ a block of $\mathscr{R}$.

\begin{definition}
  Let $\mathscr{P}$ be an equivalence relation included in a preorder
  $\mathscr{R}$. We assume that for each block $E$ of $\mathscr{P}$, a 
  representative state $E.\Rep$ has been chosen. Let $E$ be a block of
  $\mathscr{P}$, $B$ be a block of $\mathscr{R}$ and $B'\subseteq B$.  We
  define:  
  \begin{multline}
   \label{eq:RelCount}   
   \OarXiv{\hfill} \RelCount_{(\mathscr{P},\mathscr{R})}(E,B')\triangleq \Olics{\\}
    |\{E'\in P_{\mathscr{P}}\suchthat E.\Rep\rightarrow E'\wedge
    B'\mathrel\mathscr{R} E'\}| \OarXiv{\hfill}
  \end{multline}
\end{definition}

\begin{proposition}
  \label{prop:RelCount}
  Let $\mathscr{P}$ be an equivalence relation included in a preorder
  $\mathscr{R}$, $E$ be a block of $\mathscr{P}$, $B$ be a block of
  $\mathscr{R}$ and $B'$ be a non empty subset of $B$. Then:  
  \begin{displaymath}
    \begin{split}
      \RelCount_{(\mathscr{P},\mathscr{R})}(E,B)=
      \RelCount_{(\mathscr{P},\mathscr{R})}(E,B')
    \end{split}
\end{displaymath}
\end{proposition}
\begin{proof}
  Thanks to the transitivity of $\mathscr{R}$.
\end{proof}

Following Section~\ref{sec:ideas}, the purpose of these counters is
to check in constant time whether a block $E$ of an equivalence relation
$\mathscr{P}$ is included in $\rightarrow^{-1}\mathrel{\circ}\mathscr{R}(b)$ for
a given state $b$. But this is correct only if $\mathscr{P}$ is already
$\mathscr{R}$-block-stable. If this is not the case, its underlying partition
should be split accordingly.
We thus introduce the first condition which
necessitates a split of the current equivalence relation $\mathscr{P}$ to
approach the coarsest $\mathscr{R}$-block-stable equivalence relation.
For this, we take advantage of the existence of 
$\RelCount_{(\mathscr{P},\mathscr{R})}$. 

\begin{definition}
  Let $\mathscr{P}$ be an equivalence relation included in a preorder
  $\mathscr{R}$,  $E$ be a block of $\mathscr{P}$ and $B$ be a block of
  $\mathscr{R}$ such that $E\rightarrow B$ and
  $\RelCount_{(\mathscr{P},\mathscr{R})}(E,B)=0$. The transition $E\rightarrow
  B$ is called a $(\mathscr{P},\mathscr{R})$-\emph{splitter transition of type
    1}.
\end{definition}

The intuition is as follows. With a block $E$ of $\mathscr{P}$ and a block $B$
of $\mathscr{R}$, if $\mathscr{P}$ was $\mathscr{R}$-block-stable, with
$E\rightarrow B$ and Lemma~\ref{lem:partStability} we would have had
$E\,\subseteq\,\rightarrow^{-1}\mathrel\circ\mathscr{R}(B)$. But
$\RelCount_{(\mathscr{P},\mathscr{R})}(E,B)=0$ denies this. So we have to split
$P_\mathscr{P}$.

\begin{lemma}
  \label{lem:split1}
  Let $E\rightarrow B$ be a $(\mathscr{P},\mathscr{R})$-splitter transition of type
  1. Let
  $P'=\Split(P_{\mathscr{P}},\rightarrow^{-1}\mathrel\circ\mathscr{R}(B))$. Then
  $\mathscr{P}_{P'}$ is strictly included in $\mathscr{P}$ and contains all
  $\mathscr{R}$-block-stable equivalence relations included in $\mathscr{P}$.  
\end{lemma}

% \begin{figure}[!t]
%   \centering
%    \includegraphics[scale=1]{tikz/proofSplit1}
%   %   \input{proofSplit1.pgf}
%   \caption{Proof of Lemma~\ref{lem:split1}.}
%   \label{fig:proofSplit1}
% \end{figure}
\begin{proof}%[Proof of Lemma~\ref{lem:split1}]
  To prove the strict inclusion, let us show that $E$ is split.  From
  $E\rightarrow B$ there is $e\in E$ such that $e\rightarrow B$ and thus $e\in
  \rightarrow^{-1}\mathrel\circ\mathscr{R}(B)$ since $\mathscr{R}$ is
  reflexive. Furthermore, by definition,
  $\RelCount_{(\mathscr{P},\mathscr{R})}(E,B)=0$ means that
  $E.\Rep\not\in\rightarrow^{-1}\mathrel\circ\mathscr{R}(B)$. This implies that
  $e$ and $E.\Rep$ do not belong to the same block of $P'$. Thus, at least $E$
  has been split. Now, let $\mathscr{P}''$ be a $\mathscr{R}$-block-stable
  equivalence relation included in $\mathscr{P}$. If
  $\mathscr{P}''$ is not included in $\mathscr{P}'$ there are two states, $h_1$
  and $h_2$, from a block $H$ of $\mathscr{P}''$ such that $h_1\in
  \rightarrow^{-1}\mathrel\circ\mathscr{R}(B)$ and $h_2\not\in
  \rightarrow^{-1}\mathrel\circ\mathscr{R}(B)$. But, by definition, if
  $\mathscr{P}''$ is $\mathscr{R}$-block-stable,  $h_1\in
  \rightarrow^{-1}\mathrel\circ\mathscr{R}(B)$ implies  $h_2\in
  \rightarrow^{-1}\mathrel\circ\mathscr{R}(B)$ which leads to a contradiction. 
  Therefore $\mathscr{P}''$ is
  included in $\mathscr{P}'$.  
\end{proof}

Saying that $\mathscr{P}_{P'}$ is strictly included in $\mathscr{P}$ means that
at least one block of $\mathscr{P}$, here $E$, has been split to obtain $P'$.

% \begin{figure}[!t]
%   \centering
%    \includegraphics[scale=1]{tikz/afterSplit1}
%     % \input{afterSplit1.pgf}
%   \caption{When there  is no splitter transition of type 1.}
%   \label{fig:afterSplit1}
% \end{figure}
\begin{lemma}
%\begin{lemma}[\Figure\ref{fig:afterSplit1}]
  \label{lem:afterSplit1}
  Let $\mathscr{P}$ be an equivalence relation included in a preorder $\mathscr{R}$
  such that there is  no
  $(\mathscr{P},\mathscr{R})$-splitter transition of type 1.  Let $E$
  be a block of $\mathscr{P}$ and $B$ be a block of $\mathscr{R}$ such that
  $E\rightarrow B$. Then, 
  $E.\Rep\in\rightarrow_{\mathscr{R}}^{-1}\mathrel\circ\mathscr{R}(B)$.
\end{lemma}

\begin{proof}%[Proof of Lemma~\ref{lem:afterSplit1}]
  Let $E\rightarrow B$ be a transition with $E$ and $B$ satisfying the
  hypotheses of the lemma. Since there is no
  $(\mathscr{P},\mathscr{R})$-splitter transition of type 1,
  $\RelCount_{(\mathscr{P},\mathscr{R})}(E,B)\neq 0$. Therefore, $E.\Rep\in
  \rightarrow^{-1}\mathrel\circ\mathscr{R}(B)$ and thus $E.\Rep\in
  \rightarrow_{\mathscr{R}}^{-1}\mathrel\circ\mathscr{R}(B)$ by
  Lemma~\ref{lem:transInMaxoR}.
\end{proof}

Now, for $E$ to be really a representative, we need the following implication:
$E.\Rep\rightarrow B \Rightarrow E\,\subseteq\,
\rightarrow^{-1}\mathrel\circ\mathscr{R}(B)$. But to check this property
effectively, taking advantage of the counters, we need a stronger property
equivalent with the one, \eqref{eq:sameBlock}, defining block-stability. 

\begin{lemma}
  \label{lem:strongerBlockStability}
  Let $\mathscr{P}$ be an equivalence relation included in a preorder
  $\mathscr{R}$. Then \eqref{eq:sameBlock} is equivalent with: 
  \begin{equation}
    \label{eq:strongerBlockStability}
    \mathscr{P}\mathrel\circ\rightarrow_{\mathscr{R}}^{-1}\,\subseteq\,
    \rightarrow^{-1}\mathrel\circ  [\cdot]_{\mathscr{R}}
  \end{equation}  
\end{lemma}

\begin{proof}%[Proof of Lemma~ \ref{lem:strongerBlockStability}]
  With Lemma~\ref{lem:partStability} it suffices to show the equivalence between
  \eqref{eq:partStability} and \eqref{eq:strongerBlockStability}:
  
  \begin{itemize}
  \item []\eqref{eq:partStability} $\Rightarrow$
    \eqref{eq:strongerBlockStability}. Let $(c',d)\in
    \mathscr{P}\mathrel\circ\rightarrow_{\mathscr{R}}^{-1}$. There is a state
    $c$ such that  $(c,d)\in\mathscr{P}$ and
    $c\rightarrow_{\mathscr{R}}c'$. With \eqref{eq:partStability} there is a state
    $d'$ such that $c'\,\mathscr{R}\,d'$ and $d\rightarrow d'$. With
    $(d,c)\in\mathscr{P}$, since $\mathscr{P}$ is symmetric, and
    \eqref{eq:partStability} again there is a state $c''$ such that
    $d'\,\mathscr{R}\,c''$ and $c\rightarrow c''$. But $c\rightarrow c'$ being a
    $\mathscr{R}$-maximal transition implies that $c''\in[c']_{\mathscr{R}}$. So
    we have $c''\,\mathscr{R}\,c'$, $c'\,\mathscr{R}\,d'$,
    $d'\,\mathscr{R}\,c''$ and thus $d'\in[c']_{\mathscr{R}}$. This means that
    $d\in \rightarrow^{-1}\mathrel\circ [\cdot]_{\mathscr{R}}(c')$ and thus
    $(c',d)\in \rightarrow^{-1}\mathrel\circ [\cdot]_{\mathscr{R}}$. Therefore $
    \mathscr{P}\mathrel\circ\rightarrow_{\mathscr{R}}^{-1}\,\subseteq\,
    \rightarrow^{-1}\mathrel\circ [\cdot]_{\mathscr{R}}$.

  \item [] \eqref{eq:strongerBlockStability} $\Rightarrow$
    \eqref{eq:partStability}. With Lemma~\ref{lem:transInMaxoR} we have
    $\mathscr{P}\mathrel\circ\rightarrow^{-1}\,\subseteq\,\mathscr{P}\mathrel\circ\rightarrow_{\mathscr{R}}^{-1}\mathrel\circ\mathscr{R}$.
    With \eqref{eq:strongerBlockStability} we have
    $\mathscr{P}\mathrel\circ\rightarrow_{\mathscr{R}}^{-1}\mathrel\circ\mathscr{R}\,\subseteq\,
    \rightarrow^{-1}\mathrel\circ
    [\cdot]_{\mathscr{R}}\mathrel\circ\mathscr{R}$. But $\mathscr{R}$ being a
    preorder we have
    $[\cdot]_{\mathscr{R}}\mathrel\circ\mathscr{R}\,\subseteq\,\mathscr{R}$. With
    all of this: $\mathscr{P}\mathrel\circ\rightarrow_{\mathscr{R}}^{-1}\,\subseteq\,
    \rightarrow^{-1}\mathrel\circ \mathscr{R}$.
  \end{itemize}
\end{proof}

\begin{definition}
  \label{def:split2}
  Let $\mathscr{P}$ be an equivalence relation included in a preorder
  $\mathscr{R}$. Let $E$ be a block of $\mathscr{P}$ and $B$ be a block
  of $\mathscr{R}$ such that $E.\Rep\rightarrow B$,
 % $\RelCount_{(\mathscr{P},\mathscr{R})}(E,B)\neq 0$,
  $\RelCount_{(\mathscr{P},\mathscr{R})}(E,B)= |\{[b]_{\mathscr{P}}\subseteq
  B\suchthat E.\Rep\rightarrow b\}|$ and $E\nsubseteq\rightarrow^{-1}(B)$. The
  transition $E\rightarrow B$ is called a
  $(\mathscr{P},\mathscr{R})$-\emph{splitter transition of type 2}.
\end{definition}

\begin{remark}
  The conditions in Definition~\ref{def:split2} are inspired from those used in
  \cite{Ran14} for its split step.
\end{remark} 

The intuition is as follows. If $E.\Rep\rightarrow B$ and the condition on the
counter is true (all transitions from $E.\Rep$ that reach states greater,
relatively to $\mathscr{R}$, than $B$ actually have their destination states in
$B$), this means that the transition $E.\Rep\rightarrow B$ is maximal. With
Lemma~\ref{lem:strongerBlockStability}, if $\mathscr{P}$ is
$\mathscr{R}$-block-stable this should imply
$E\subseteq\rightarrow^{-1}(B)$. Since this is not the case, $\mathscr{P}$ must
be split.

\begin{lemma}
  \label{lem:split2}
  Let $\mathscr{P}$ be an equivalence relation included in a preorder $\mathscr{R}$
  such that there is no $(\mathscr{P},\mathscr{R})$-splitter transition of type
  1 and let $E\rightarrow B$ be a splitter transition of type 2.
 Let
  $P'=\Split(P_{\mathscr{P}},E\cap\rightarrow^{-1}(B))$. Then
  $\mathscr{P}_{P'}$ is strictly included in $\mathscr{P}$ and contains all
  $\mathscr{R}$-block-stable equivalence relations included in $\mathscr{P}$.
\end{lemma}

\begin{figure}[!t]
  \centering
   \includegraphics[scale=1]{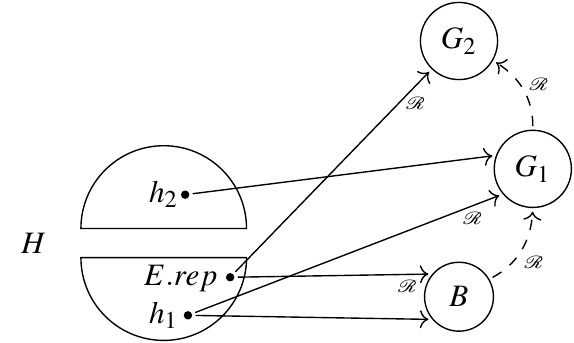}
   \caption{Proof of second statement of Lemma~\ref{lem:split2}.}
   \label{fig:proofSplit2}   
\end{figure}
\begin{proof}%[Proof of Lemma~\ref{lem:split2}]
  To prove the strict inclusion, let us show that $E$ is split.  From
  $E\rightarrow B$ being a $(\mathscr{P},\mathscr{R})$-splitter transition of
  type 2 we have $|\{[b]_{\mathscr{P}}\subseteq B\suchthat E.\Rep\rightarrow
  b\}|\neq 0$ and $E\nsubseteq\rightarrow^{-1}(B)$. The first property implies
  that $ E.\Rep\in {\rightarrow^{-1}(B)}$ and the second one implies the existence
  of a state $e\in E$ which does not belong to $\rightarrow^{-1}(B)$. Therefore
  $E$ has been split.  The second statement is proved by contradiction. Let
  $\mathscr{P}''$ be a $\mathscr{R}$-block-stable equivalence relation included in
  $\mathscr{P}$. Consider \Figure\ref{fig:proofSplit2}. If $\mathscr{P}''$ is
  not included in $\mathscr{P}'$ there are two states, $h_1$ and $h_2$, from a
  block $H$ of $\mathscr{P}''$ such that $h_1\in \rightarrow^{-1}(B)$ and
  $h_2\not\in \rightarrow^{-1}(B)$.  With Lemma~\ref{lem:transInMaxoR} there is
  a block $G_1$ of $\mathscr{R}$ such that: $B\mathrel\mathscr{R}G_1$ and
  $h_1\rightarrow_{\mathscr{R}}G_1$. Since $\mathscr{P}''$ is
  $\mathscr{R}$-block-stable this implies $h_2\rightarrow G_1$. With
  Lemma~\ref{lem:afterSplit1} there is a block $G_2$ of $\mathscr{R}$ such that
  $E.\Rep\rightarrow_{\mathscr{R}} G_2$ and $G_1\mathrel\mathscr{R} G_2$. But
  the condition $\RelCount_{(\mathscr{P},\mathscr{R})}(E,B)=
  |\{[b]_{\mathscr{P}}\subseteq B\suchthat E.\Rep\rightarrow b\}|$ implies that
  the transition from $E.\Rep$ to $B$ is maximal which implies that
  $G_2=B=G_1$. So we have $h_2\in \rightarrow^{-1}(B)$ which contradicts an
  above assumption. Therefore $\mathscr{P}''$ is included in $\mathscr{P}'$.  
\end{proof}

\begin{theorem}
  \label{th:split1_2}
 Let $\mathscr{P}$ be an equivalence relation included in a preorder $\mathscr{R}$
 such that there is no
  $(\mathscr{P},\mathscr{R})$-splitter transition of type 1 or
  $(\mathscr{P},\mathscr{R})$-splitter transition of type 2. Then $\mathscr{P}$
  is $\mathscr{R}$-block-stable.
\end{theorem}
\begin{figure}[!t]
  \centering
  \includegraphics[scale=1]{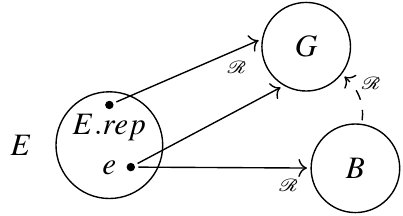}

   \caption{Proof of Theorem~\ref{th:split1_2}.}
   \label{fig:proofThSplit12partStable}   
\end{figure}
\begin{proof}%[Proof of Theorem~\ref{th:split1_2}]
  Consider \Figure\ref{fig:proofThSplit12partStable}. Let us consider a
  transition $E\rightarrow_{\mathscr{R}} B$ with $E$ a block of $\mathscr{P}$
  and $B$ a block of $\mathscr{R}$. By definition there is a state $e\in E$ such
  that $e\rightarrow_{\mathscr{R}} B$ and by Lemma~\ref{lem:afterSplit1} there
  is a block $G$ such that $B\mathrel\mathscr{R} G$ and
  $E.\Rep\rightarrow_{\mathscr{R}} G$. From ${E.\Rep\rightarrow_{\mathscr{R}}
    G}$ it is easy to show that $\RelCount_{(\mathscr{P},\mathscr{R})}(E,G)\neq
  0$ and $\RelCount_{(\mathscr{P},\mathscr{R})}(E,G)=
  |\{[g]_{\mathscr{P}}\subseteq G\suchthat E.\Rep\rightarrow g\}|$ since the
  transition $E.\Rep\rightarrow G$ is maximal. But since there is neither
  $(\mathscr{P},\mathscr{R})$-splitter transition of type 1 nor
  $(\mathscr{P},\mathscr{R})$-splitter transition of type 2 then
  $E\subseteq\rightarrow^{-1}(G)$ and thus $e\rightarrow G$. With
  $e\rightarrow_{\mathscr{R}} B$ and $B\mathrel\mathscr{R} G$ we necessarily get
  $B=G$ and thus $E\subseteq\rightarrow^{-1}(B)$. So we have
  $E\rightarrow_{\mathscr{R}} B$ implies $E\subseteq\rightarrow^{-1}(B)$. This is
  equivalent of saying that  \eqref{eq:strongerBlockStability} is true. With
  Lemma~\ref{lem:strongerBlockStability} this implies that $\mathscr{P}$ is
  $\mathscr{R}$-block-stable.
  
\end{proof}

Therefore, in the algorithm, before a refine step on $\mathscr{R}$ using
Theorem~\ref{th:refinement}, we will start from the partition $P_{\mathscr{R}}$
and split it in conformity with Lemma~\ref{lem:split1} and
Lemma~\ref{lem:split2}. By doing so, we will obtain the coarsest
$\mathscr{R}$-block-stable equivalence relation.

The next proposition shows where to search splitter transitions: those who ends
in blocks $B$ of $\mathscr{R}$ such that $\mathit{NotRel}(B)$ is not empty.

\begin{proposition}
  \label{prop:notRel}   
  Let $\mathscr{U}$ be a preorder, $\mathscr{R}$ be a $\mathscr{U}$-stable
  preorder, $\mathscr{P}$ be an equivalence relation included in $\mathscr{R}$ and
  let $\mathit{NotRel}=\mathscr{U}\mathrel\setminus\mathscr{R}$. Then   
  
  \begin{enumerate}
  \item If $E\rightarrow B$ is a $(\mathscr{P},\mathscr{R})$-splitter
    transition of type 1 then $\mathit{NotRel}(B)\neq \emptyset$. 
  \item Under the absence of $(\mathscr{P},\mathscr{R})$-splitter
    transition of type 1, if $E\rightarrow B$ is a
    $(\mathscr{P},\mathscr{R})$-splitter transition of type 2 then
    $\mathit{NotRel}(B)\neq \emptyset$.     
  \end{enumerate}
\end{proposition}

\begin{figure}[!t]
  \centering
   \includegraphics[scale=1]{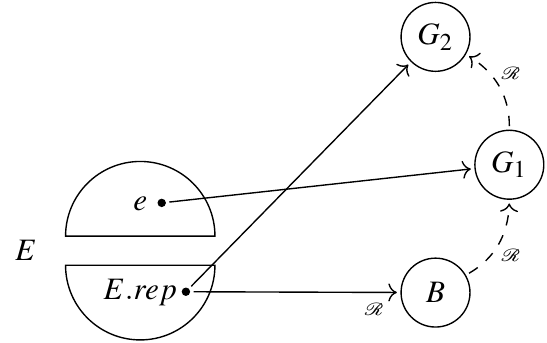}
  \caption{Proof of second item of Proposition~\ref{prop:notRel}.}
  \label{fig:notRel}   
\end{figure}

\begin{proof}%[Proof of Proposition~\ref{prop:notRel}]
  \mbox{}
  \begin{enumerate}
  \item Let $E\rightarrow B$ be a $(\mathscr{P},\mathscr{R})$-splitter
    transition of type 1. By definition there is $e\in E$ such that
    $e\rightarrow B$. Since $E$ is a block of $\mathscr{P}$ we have
    $e\mathrel\mathscr{P} E.\Rep$ and thus $e\mathrel\mathscr{R} E.\Rep$.  With
    the hypothesis that $\mathscr{R}$ is $\mathscr{U}$-stable we get
    $E.\Rep\in\rightarrow^{-1}\mathrel\circ\mathscr{U}(B)$.  But the hypothesis
    that $E\rightarrow B$ is a $(\mathscr{P},\mathscr{R})$-splitter transition
    of type 1 implies
    $E.\Rep\not\in\rightarrow^{-1}\mathrel\circ\mathscr{R}(B)$. From these two
    last constraints on $E.\Rep$ we get
    $E.\Rep\in\rightarrow^{-1}\mathrel\circ(\mathscr{U}\mathrel\setminus\mathscr{R})(B)$
    and thus $\mathit{NotRel}(B)\neq \emptyset$.
    
  \item Consider \Figure\ref{fig:notRel}. Let $E\rightarrow B$ be a
    $(\mathscr{P},\mathscr{R})$-splitter transition of type 2. By definition we
    get $E.\Rep\rightarrow_{\mathscr{R}}B$ and
    $E_2=E\mathrel\setminus(\rightarrow^{-1}(B))$ is not empty. Let $e\in E_2$,
    with a similar argument than the first item, we get
    $e\in\rightarrow^{-1}\mathrel\circ\mathscr{U}(B)$. By contradiction, let us
    assume $e\in\rightarrow^{-1}\mathrel\circ\mathscr{R}(B)$. There is $G_1$ a
    block of $\mathscr{R}$ such that $B\mathrel\mathscr{R} G_1$ and
    $e\rightarrow G_1$. From Lemma~\ref{lem:afterSplit1} there is $G_2$ a block
    of $\mathscr{R}$ such that $E.\Rep\rightarrow G_2$ and
    $G_1\mathrel\mathscr{R} G_2$, and thus $B\mathrel\mathscr{R} G_2$. Since the
    transition from $E.\Rep$ to $B$ is maximal this implies that $B=G_1=G_2$
    which contradicts $e\not\in\rightarrow^{-1}(B)$. Therefore
    $e\not\in\rightarrow^{-1}\mathrel\circ\mathscr{R}(B)$. With
    $e\in\rightarrow^{-1}\mathrel\circ\mathscr{U}(B)$ we get
    $e\in\rightarrow^{-1}\mathrel\circ(\mathscr{U}\mathrel\setminus\mathscr{R})(B)$
    and thus $\mathit{NotRel}(B)\neq \emptyset$.    
  \end{enumerate}
\end{proof}

We have now everything to propose an efficient algorithm.

\subsection{Data Structures and Space Complexity}
\label{sec:data-structures}

% To avoid pathological case, we assume there is no isolated state. So each state
% has at least one input or output transition. 

\begin{remark}
  Since the final partition $P_{\mathrm{sim}}$ is obtained after several splits
  of the initial partition $P_{\mathrm{init}}$ we have $|P|\leq
  |P_{\mathrm{sim}}|$ with $P$ the current partition at a given step of the
  algorithm.
\end{remark}

The current relation $\mathscr{R}$ is represented in the algorithm by a
partition-relation pair $(P,Rel)$.  The data structure used to represent a
partition is traditionally a (possibly) doubly linked list of blocks, themselves
represented by a doubly linked list of states.  But in practice, each node of a
list contains a reference, to the next node of the list. The size (typically 64
bits nowadays) of these references is static and does not depend on the size of
the list. We therefore prefer the use of arrays because we can control
the size of the slots and manipulate arrays is faster than manipulate lists. The
idea, see \cite{VL08} and \cite{Cec13a} for more details, is to identify a state with
its index in $Q$ and to distribute these indexes in an array, let us name it
$T$, such that states belonging to the same block are in consecutive slots in
$T$.  A block could thus be represented by the indexes of its first and last
elements in $T$, the two indexes defining a contiguous subarray. If a block is
split, the two subblocks form two contiguous subarrays of that subarray.  By
playing with these arrays (other arrays are needed, like the one giving the
position in $T$ of a state) we obtain a representation of a partition which
allows splitting (some elements of a block are removed from their original block
and form a new block) and scanning of a block in linear time.

However, as seen in the previous sections, we need two
generations of blocks at the same time: the first one corresponds to blocks of
$\mathscr{R}$ and the second one corresponds to blocks of the next generation,
$\mathscr{V}$, of this preorder. Hence, we need an intermediate between blocks
and their corresponding states: \emph{nodes}. A node corresponds to a block or
to an ancestor of a block in the family tree of the different generations of
blocks issued from the split steps of the algorithm. To simplify the writing, we
associate in the present paper, a node to the set of its corresponding states. As
an example, consider 
a block $B$ which consists of the following three states $\{q_1,q_2,q_3\}$. In
reality, we will associate $B$ to a node $N=\{q_1,q_2,q_3\}$. In this way, if
$B$ is split in $B_1$ and $B_2$, corresponding respectively to $\{q_1,q_2\}$ and
$\{q_3\}$, we create two new nodes $N_1=\{q_1,q_2\}$ and $N_2=\{q_3\}$ and we
associate $B_1$ to $N_1$ and $B_2$ to $N_2$. By doing so, $N$ remains bounds to
the set $\{q_1,q_2,q_3\}$.  To represent a node we just need to keep in memory
the index of its first element and the index of its last element in the array
$T$ (see the previous paragraph). When a block which corresponds to a node is
split, the corresponding states change their places in $T$ but keep in the same
subarray. Let us note that when a block is split, it is necessarily in two parts.
Therefore, the number of nodes is at most twice the number of blocks and the bit
space needed to represent the partition and the nodes is in $O(|Q|.\log(|Q|))$
since there is less blocks than states.

To be more precise about the relations between states, blocks and nodes: at any
time of the algorithm, the index of a state $q$ is associated to the index of
its block $q.\Block$, the index of a block $E$ is associated to the index of its
node $E.\Node$ and the index of a node is associated to the states it contains
(via two indexes of the array $T$). A node which is not linked by a block is an
ancestor of, at least two, blocks.  By the data structure chosen to represent
the partition it is easy to see that given a node $N$ we can scan in linear time
the states it contains (this corresponds to the scan of a contiguous subarray)
and the blocks it contains (by a similar process). The function
$\ChooseBlock(N)$ which arbitrarily choose one block whose set of elements is
included in those of $N$ is executed in constant time (we choose $e.\Block$,
with $e$ the first element of $N$). Similarly, the function $\ChooseState{E}$
which returns a state of a block $E$, used to defined $E.\Rep$ a representative
of $E$, is also executed in constant time (we choose the first element in
$E.\Node$).

The relation $Rel$ is distributed on the blocks. To each block $C$ we
associate an array of booleans, $C.\Rel$, such that $(C,D)\in Rel$, what we note $D\in C.\Rel$, iff the boolean
at the index of the block $D$ in $C.\Rel$ is true. These arrays are resizable
arrays whose capacities are doubled as needed. Therefore the classical
operations on these arrays, like get and set, take constant amortized time. We
use this type of array wherever necessary. The bit size needed to represent
$Rel$ is therefore in $O(|P_{\mathrm{sim}}|^2)$.

The relation $\mathscr{U}$ that appears in the previous sections is not directly
represented. We use instead the equality $\mathscr{U} = \mathscr{R}\cup \mathit{NotRel}$
and represent $\mathit{NotRel}$. Since $\mathscr{U}$ is a coarser preorder than
$\mathscr{R}$, for a given node $B$ which represents a block of $\mathscr{R}$,
the set $\mathit{NotRel}(B)$ is represented in the algorithm by $B.\NotRel$ a set of
nodes (encoded by a resizable array of the indexes of the corresponding nodes)
which represent blocks of $\mathscr{R}$. As explained earlier, we have to use
nodes instead of blocks because nodes never change whereas blocks can be split
afterwards.  The bit space representation of $\mathit{NotRel}$ is thus in
$O(|P_{\mathrm{sim}}|^2.\log(|P_{\mathrm{sim}}|))$. Remember, the number of
nodes is linear in the maximal number of blocks: $|P_{\mathrm{sim}}|$.

In Section~\ref{sec:splitter-transitions} we introduced a counter,
$\RelCount_{(\mathscr{P},\mathscr{R})}(E,B)$, for each pair made of a block $E$
of the current equivalence relation $\mathscr{P}$ represented in the algorithm by
the current partition $P$, and a block $B$ of the current relation $\mathscr{R}$
represented in the algorithm by $(P,Rel)$. As seen in
Proposition~\ref{prop:RelCount}, for any subblock $B'\in P_{\mathscr{P}}$ of $B$
we have $\RelCount_{(\mathscr{P},\mathscr{R})}(E,B) =
\RelCount_{(\mathscr{P},\mathscr{R})}(E,B')$. Therefore, we can limit these
counters to any pair of blocks of the current partition $P$ in the
algorithm. Such a counter counts a number of blocks. This means that the total
bit size of these counters is in
$O(|P_{\mathrm{sim}}|^2.\log(|P_{\mathrm{sim}}|))$. In practice, we associate to
each block $B'$ a resizable array, $B'.\RelCountAlgo$, of $|P|$ elements such that
$B'.\RelCountAlgo(E)=\RelCount_{(\mathscr{P},\mathscr{R})}(E,B')$.

At several places in the algorithm we use a data structure to manage a set of
indexed elements. This is the case for $Touched$, $Touched'$, $Marked$,
$RefinerNodes$, $RefinerNodes'$, $PreB'$, $Remove$, $PreE$ and $PreE'$.  Such a
set is implemented by a resizable boolean array, to know in constant time whether a given
element belongs to the set, and another resizable array to store the indexes of
the elements which belongs to the set. This last array serves to scan in linear
time the elements of the set or to emptied the set in linear time of the number
of the elements in the set. So the operations, add an element in the set and
test whether an element belongs to the set are done in constant time or
amortized constant time. Also, scanning the elements of the set and emptying the
set are executed in linear time of the size of the set.  We use a finite number
of these sets for states, blocks or nodes. The overall bit space used for them
is therefore in $O(|Q|.\log(|Q|))$.  The other variables used in the algorithm,
some booleans and a counter, $E.\Count$, associated to each block $E$ in
Function \ref{func:Split2} are manipulated by constant time operations and need a
bit space in $O(|Q|.\log(|Q|))$ since $|P|$ and the number of nodes are both in
$O(|Q|)$. From all of this, we derive the following theorem.

\begin{theorem}
  The overall bit space used by the presented simulation algorithm is in \OarXiv{\\}
  $O(|P_{\mathrm{sim}}|^2.\log(|P_{\mathrm{sim}}|) + |Q|.\log(|Q|))$. 
\end{theorem}
\begin{remark}
  We assume one can iterate through the transition relation $\rightarrow$ in
  linear time of its size.  From each state $q$ we also assume one can iterate
  through the set $\rightarrow^{-1}(q)$ in linear time of its size. It is a
  tradition in most articles dealing with simulation to not count the space used
  to represent the transition relation since it is considered as an input
  data. If it was to be counted it would cost $O(|\rightarrow|.\log(|Q|))$ bits.
\end{remark}

\subsection{Procedures and Time Complexity}

In this section we analyze the different functions of the algorithm and give
their overall time complexities. The reader should remember that, in the
algorithm, a block is just an index and the set of states corresponding to a block $E\in P$ is $E.\Node$.

\subsubsection{Function \ref{func:Sim}}
\label{sec:Sim}

\begin{function}[!t]
  \caption{Sim($Q,\rightarrow,P_{\mathrm{init}}, R_{\mathrm{init}}$)}
  \label{func:Sim}
  \KwData{$RefinerNodes$: the set of nodes $B$ corresponding to blocks of 
    $\mathscr{R}$ such that $B.\NotRel\neq\emptyset$}
  $RefinerNodes:=\emptyset$ \;
  
  $P:=$\texttt{\ref{func:Init}($Q,\rightarrow,P_{\mathrm{init}},R_{\mathrm{init}}, RefinerNodes$)} \;
  \BlankLine
  \While{$RefinerNodes\neq\emptyset$\nllabel{Sim:3}}
  {
    \texttt{\ref{func:SimUpdateData}($P, RefinerNodes$)} \nllabel{Sim:4}\;

    \BlankLine
    \texttt{\ref{func:Split1}($P, RefinerNodes$)} \;
    \texttt{\ref{func:Split2}($P, RefinerNodes$)} \;
    \texttt{\ref{func:Refine}($RefinerNodes$)} \;
  }
  \BlankLine
  $P_{\mathrm{sim}} := P$ ; $R_{\mathrm{sim}}:= \{(C,D)\in P\times P\suchthat D\in C.\Rel\}$ \;
  \Return{$(P_{\mathrm{sim}}$, $R_{\mathrm{sim}})$}
\end{function}

This is the main function of the algorithm. It takes as input a transition
system $(Q,\rightarrow)$ and an initial preorder, $\mathscr{R}_{\mathrm{init}}$,
represented by the partition-relation pair
$(P_{\mathrm{init}},R_{\mathrm{init}})$.  Let us define the two following
relations:

\begin{align}
  \label{eq:defRRel}
  \mathscr{R}&\triangleq\bigcup_{\{(E,E')\in P^2\suchthat E'\in E.\Rel\}}E.\Node\times E'.\Node\\
  \label{eq:defNotRelAlgo}
  \mathit{NotRel}&\triangleq\bigcup_{B\;\in\; RefinerNodes} B\times (\cup B.\NotRel)
\end{align}

Let $\mathscr{R}_0=Q\times Q$ and $\mathscr{R}_i$ (resp. $\mathit{NotRel}_i$) be the value of
$\mathscr{R}$  (resp. $\mathit{NotRel}$) at the $i^{\text{th}}$ iteration of the while loop at line
\ref{Sim:3} of Function \ref{func:Sim}. We will show (in the analysis of Function
\ref{func:Init} for the base case and procedures \ref{func:SimUpdateData} and
\ref{func:Refine} for the inductive step) that, at this line
\ref{Sim:3}, we maintain the five following properties at the $i^{\text{th}}$
iteration of the while loop:
\begin{equation}\label{eq:RelNotRelStability}
 \mathscr{R}_{i}\text{ is } \mathscr{R}_{i-1}\text{-stable}
\end{equation}
\begin{equation}\label{eq:SimIncludedRecc}
 \text{A simulation included in }\mathscr{R}_{i-1}\text{ is included in }\mathscr{R}_{i} 
\end{equation}

\begin{equation} \label{eq:RiDecomposition}
  \mathit{NotRel}_{i}=\mathscr{R}_{i-1}\mathrel\setminus\mathscr{R}_{i} \text{ and thus }
  \mathscr{R}_{i-1} = \mathscr{R}_{i} \cup  \mathit{NotRel}_{i}
\end{equation}
\begin{multline}
  \label{eq:RefinerNodes}
  \hfill\text{$RefinerNodes$ is the set of nodes $B$ corresponding to}\hfill \\
\hfill  \text{blocks of $\mathscr{R}_i$ such that $B.\NotRel\neq\emptyset$}\hfill
\end{multline}
  \begin{multline}\label{eq:SplitUpdateData}
     \hfill  \forall E,B'\in P\;.\; B'.\RelCountAlgo{E} =
      |\{E'\in P\suchthat E.\Rep\rightarrow  \Olics{\hfill\\  \hfill} E'.\Node \wedge
      B'.\Node\times E'.\Node \subseteq \mathrel{\mathscr{R}_{i-1}} \}|\hfill
\end{multline}

From \eqref{eq:RelNotRelStability}, and thus
$\mathscr{R}_{i}\subseteq\mathscr{R}_{i-1}$, \eqref{eq:RiDecomposition},
\eqref{eq:RefinerNodes} and the condition of the while loop we get that
$(\mathscr{R}_i)_{i\geq 0}$ is a strictly decreasing sequence of
relations. Since the underlying set of states is finite, this sequence reaches a
limit in a finite number of iterations. Furthermore, if this limit is reached at
the $k^{\text{th}}$ iteration, then, from \eqref{eq:RefinerNodes} and the
condition of the while loop, we have $\mathit{NotRel}_{k+1}=\emptyset$ and from
\eqref{eq:RiDecomposition} and \eqref{eq:RelNotRelStability} we obtain that
$\mathscr{R}_{k}=\mathscr{R}_{k+1}$ and $\mathscr{R}_{k+1}$ is
$\mathscr{R}_{k}$-stable. Which means that $\mathscr{R}_{k}$ is a
simulation. From \eqref{eq:SimIncludedRecc} and the fact, to be shown in the
analysis of function Init, that all simulation included in
$\mathscr{R}_{\mathrm{init}}$ is also included in $\mathscr{R}_{1}$, we deduce
that $\mathscr{R}_{k}$ contains all simulation included in
$\mathscr{R}_{\mathrm{init}}$. Therefore, Function \ref{func:Sim} returns a
partition-relation pair that corresponds to $\mathscr{R}_{\mathrm{sim}}$, the
coarsest simulation included in $\mathscr{R}_{\mathrm{init}}$.
%\gnote{A relire}

The fact that $(\mathscr{R}_i)_{i\geq 0}$ is a strictly decreasing sequence of
relations and \eqref{eq:RiDecomposition} imply the following lemma that will be
used as a key argument to analyze the time complexity of the algorithm.
\begin{lemma}
  \label{lem:InOnlyOne}
  Two states, and thus two blocks or two nodes, are related by $\mathit{NotRel}$ in at most
  one iteration of the while loop in function \ref{func:Sim}.
\end{lemma}

\subsubsection{Procedure \ref{func:SimUpdateData}}
\label{sec:SimUpdateData}

\begin{procedure}[!t]
  \caption{SimUpdateData($P, RefinerNodes$)}
  \label{func:SimUpdateData}
  \BlankLine
  $PreE' := \emptyset$ \;
  
  \ForEach{$B\in RefinerNodes$}
  {
    $B'=\ChooseBlock{B}$ \;
    \tcp{\small At this stage, $B'$ is the only block of $B$}
    \ForEach{$E'\in P \suchthat E'.\Node\subseteq\cup B.\NotRel$}
    {
      \ForEach{$e\in\rightarrow^{-1}(E'.\Node)$\nllabel{SimUpdateData:5}}
      {
        $E := e.\Block$ \;
        \If {$e=E.\Rep$}
        {
          $PreE' := PreE' \cup \{E\}$ \;          
        }
      }
      \lForEach{$E \in PreE'$}
      {
        $B'.\RelCountAlgo(E)--$
      }
      $PreE' := \emptyset$ \;      
    }
  }
\end{procedure}

Assuming \eqref{eq:SplitUpdateData} is true, the role of this procedure is to
render the following formula true after line~\ref{Sim:4} of Function
\ref{func:Sim} during the $i^{\text{th}}$ iteration of the while loop.  In this way,
the counters are made consistent with~\eqref{eq:RelCount}:
  \begin{multline}\label{eq:SplitUpdateData_i}
    \OarXiv{\hfill}   \forall E,B'\in P\;.\; B'.\RelCountAlgo{E} = 
      |\{E'\in P\suchthat E.\Rep\rightarrow  \Olics{\\} E'.\Node\wedge
      B'.\Node\times E'.\Node\subseteq\mathrel{\mathscr{R}_i} \}| \OarXiv{\hfill}
\end{multline}

From \eqref{eq:defNotRelAlgo}, \eqref{eq:RiDecomposition},
\eqref{eq:RefinerNodes} and \eqref{eq:SplitUpdateData} we just have, for each
node $B$ in $RefinerNodes$, to scan the blocks in $\cup B.\NotRel$ in order to
identify their predecessor blocks. The corresponding counters are then
decreased.  The lines which are the most executed are those in the loop starting
at line~\ref{SimUpdateData:5}. They are executed once for each pair of
$(e\rightarrow e',B)$ with $e'\in \cup B.\NotRel$. But from Lemma~\ref{lem:InOnlyOne}
such a pair can be considered only once during the life time of the algorithm
and thus the overall time complexity of this procedure is in
$O(|P_{\mathrm{sim}}|.|\rightarrow|)$.

\subsubsection{Procedure \ref{func:Split}}
\label{sec:Split}

\begin{procedure}[!t]
  \caption{Split($P, Marked$)}
  \label{func:Split}
  \BlankLine
  $Touched := \emptyset$ \;
  
  \lForEach{$r\in Marked$} { $Touched := Touched \cup \{r.\Block \}$ }
  
  \ForEach{$C\in Touched \suchthat C.\Node\nsubseteq Marked$}
  {
    $D:=\NewBlock$ ; $P := P \cup \{D\}$  \;
    $D.\Node := C.\Node \cap Marked$ \; 
    \lForEach{$q\in D.\Node$} {$q.\Block := D$ }
    \BlankLine
    \tcp{\small The subblock which is disjoint from $Marked$ keeps the identity of  $C$.
    }
    $C.\Node:= C.\Node  \mathrel\setminus D.\Node$ \;      
     
    \BlankLine
    \texttt{\ref{func:SplitUpdateData}($P,C,D$)} \;
  }
  $Touched := \emptyset$ \;
\end{procedure}

This procedure corresponds to Definition~\ref{def:split}. The differences are
that Procedure~\ref{func:Split} transforms the current partition (it is not a
function) and each time a block is split,  Procedure
\ref{func:SplitUpdateData} is called to update the data structures (mainly \Rel
and \RelCountAlgo).  Apart from the call of \ref{func:SplitUpdateData}, which we
discuss right after, it is known that, with a correct implementation found in
most articles from the bibliography of the present paper, a call of
\ref{func:Split} is done in $O(|Marked|)$-time. Therefore, we only give a high
level presentation of the procedure.
% \gnote{ The detailed presentation will be given in
% the full paper. En donner eventuellement une version detaillee dans les
%   annexes.}

\subsubsection{Procedure \ref{func:SplitUpdateData}}
\label{sec:func:SplitUpdateData}

\begin{procedure}[!t]
  \caption{SplitUpdateData($P,C,D$)}
  \label{func:SplitUpdateData}
  \KwData{$C$ keeps the identity of the parent block ; $D$ is the new block}
  \BlankLine
  $D.\Count := 0$ \;
  \BlankLine
  \tcp{\small Update of the \Rel's}
  $D.\Rel := \Copy{C.\Rel}$\nllabel{SplitUpdateData:2} \;
  \lForEach{$E\in P\suchthat C\in E.\Rel$\nllabel{SplitUpdateData:3}}
  {
    $E.\Rel := E.\Rel \cup \{D\}$ 
  }
  
  \BlankLine
  \tcp{\small Update of the \RelCountAlgo's}
  \BlankLine
  $D.\Rep := C.\Rep$\nllabel{SplitUpdateData:4} \;
  $D.\RelCountAlgo := \Copy{C.\RelCountAlgo}$\nllabel{SplitUpdateData:5} \;

   \If{$C.\Rep.\Block = C$\nllabel{SplitUpdateData:6}}
    {
      $X := D$ \;
    }
    \Else
    {
      \ForEach{$B'\in P$}{$B'.\RelCountAlgo{D} := B'.\RelCountAlgo{C}$ \;}
      $X := C$ \;
    }
  
  \BlankLine
  \tcp{\small Update of the \RelCountAlgo's from predecessors of both $C$ and $D$}
  $Touched := \emptyset$ ;   $Touched' := \emptyset$\nllabel{SplitUpdateData:12} \;
  \ForEach{$e\in \rightarrow^{-1}(C.\Node)\suchthat e = e.\Block.\Rep$}
  {$Touched := Touched \cup \{e.\Block\}$ ;}
  
  \ForEach{$e\in \rightarrow^{-1}(D.\Node)\suchthat {e = e.\Block.\Rep}\wedge
    {e.\Block\in Touched} \wedge e.\Block\not\in Touched'$\nllabel{SplitUpdateData:15}}
  {$Touched' := Touched' \cup \{e.\Block\}$ \nllabel{SplitUpdateData:16};} 
  
  \ForEach{$E\in Touched'$\nllabel{SplitUpdateData:17}}
  {
    \ForEach{$B'\in P\suchthat C\in B'.\Rel$\nllabel{SplitUpdateData:18}}{$B'.\RelCountAlgo{E}++ $\nllabel{SplitUpdateData:19} \;}
  }
  
  $Touched := \emptyset$ ;   $Touched' :=
  \emptyset$\nllabel{SplitUpdateData:20} \; 

  \BlankLine
  \tcp{\small Compute the \RelCountAlgo's from $X$ ($C$ or $D$) which did not
    inherited the old $C.\Rep$.}
 
  $X.\Rep := \ChooseState{X}$\nllabel{SplitUpdateData:21} \;
  \lForEach{$B'\in P$} {$B'.\RelCountAlgo{X} := 0$ }
  \ForEach{$ q \rightarrow q' \suchthat q = X.\Rep$} 
  {$Touched := Touched \cup\{q'.\Block\}$ \;}
  
  \ForEach{$E'\in Touched$\nllabel{SplitUpdateData:25}}%25
  {
    \ForEach{$B'\in P\suchthat E'\in B'.\Rel$}{$B'.\RelCountAlgo{X}++$ ;\nllabel{SplitUpdateData:27}}
  }
  $Touched := \emptyset$\nllabel{SplitUpdateData:28} \;
  $C.\Node.\NotRel:=\emptyset$ ;  $C.\Node.\NotRel':=\emptyset$ \;
  $D.\Node.\NotRel:=\emptyset$ ;  $D.\Node.\NotRel':=\emptyset$ \; 
\end{procedure}

The purpose of this procedure is to maintain the data structures coherent after
a split of a block $C$ in two subblocks, the new $C$ and a new $D$.
Since this procedure only modifies the $\Rel$'s and $\NotRel$'s, we will only look at
\eqref{eq:defRRel} and \eqref{eq:SplitUpdateData_i}.

The  \Rel's are updated at lines \ref{SplitUpdateData:2} and
\ref{SplitUpdateData:3}. Let $\mathscr{R}'$ be the value of $\mathscr{R}$ before
the split, let $\mathscr{R}''$ be its value after the split and let $N$ be the
node associated with the block $C$ before it was split. Before the split, we
have $N=C.\Node$ and after the split we have $N=C.\Node\cup D.\Node$. With line
\ref{SplitUpdateData:2}, we have, for each block $E\in P$, the equivalence $
N\times E.\Node \subseteq \mathscr{R}'\Leftrightarrow (C.\Node\cup
D.\Node)\times E.\Node\subseteq \mathscr{R}''$ and with line
\ref{SplitUpdateData:3} we have the equivalence $E.\Node\times N\subseteq
\mathscr{R}'\Leftrightarrow E.\Node\times (C.\Node\cup D.\Node)\subseteq
\mathscr{R}''$. And thus $\mathscr{R}$ has not changed since the other products
$E.\Node\times E'.\Node$ in its definition have not changed. For one call of
\ref{func:SplitUpdateData} these two lines are executed in $O(|P|)$-time. Since
\ref{func:SplitUpdateData} is called only when a block is split and since there
is, at the end, at most $|P_{\mathrm{sim}}|$ blocks. The overall time complexity
of these two lines is in $O(|P_{\mathrm{sim}}|^2)$.

The \RelCountAlgo's are updated in the other lines of the procedure. In
\eqref{eq:SplitUpdateData_i}, the split can involve three blocks: $B'$, $E$ or
$E'$. Let us remember that $\mathscr{R}$ is not changed during this procedure.

Line \ref{SplitUpdateData:5} treats the case where the split block is a
$B'$. Since $\mathscr{R}$ is a preorder, after lines \ref{SplitUpdateData:2} and
\ref{SplitUpdateData:3}, we have $C\in D.\Rel$ and $D\in C.\Rel$. It is therefore
normal that for any $E\in P$ we have $D.\RelCountAlgo{E} = C.\RelCountAlgo{E}$
since $\mathscr{R}$ is a preorder.
The overall time complexity of line \ref{SplitUpdateData:5} is thus in
$O(|P_{\mathrm{sim}}|^2)$.

If in \eqref{eq:SplitUpdateData_i} the split block is $E$, the test at line
\ref{SplitUpdateData:6} determines the block $X$ among the new $C$ and $D$ that
did not inherit $E.\Rep$ (which is $C.\Rep$ at this line since $C$ keeps the
identity of the parent block, the old $C$ and thus $E$). This means that we
will have to initialise $B'.\RelCountAlgo{X}$ for all $B'\in P$. This is done
after at lines \ref{SplitUpdateData:21} to \ref{SplitUpdateData:27}. For now,
 if $D$ has inherited $E.\Rep$ we do  $B'.\RelCountAlgo{D} :=
 B'.\RelCountAlgo{E}$ for all $B'\in P$. Remember, at this stage, we have $B'.\RelCountAlgo{E} =
 B'.\RelCountAlgo{C}$. 

 Lines \ref{SplitUpdateData:12} to \ref{SplitUpdateData:20} treat the case where
 the split block is $E'$. We thus have $E'.\Node=C.\Node\cup D.\Node$. There is
 three alternatives for a given block $E$: either $E.\Node\rightarrow C.\Node$,
 or $E.\Node\rightarrow D.\Node$, or both. For the two first alternatives
 $B'.\RelCountAlgo{E}$ does not change. But for the third one, we have to
 increment this count by one. Apart for lines \ref{SplitUpdateData:17} to
 \ref{SplitUpdateData:19} the overall time complexity of these lines is in
 $O(|P_{\mathrm{sim}}|.|\rightarrow|)$. Remember, \ref{func:SplitUpdateData} is
 called only when a block is split and this occurs at most $|P_{\mathrm{sim}}|$
 times. To correctly analyze the time complexity of lines
 \ref{SplitUpdateData:17} to \ref{SplitUpdateData:19}, for a state $e$ let us
 first define $\rightarrow(e)_P\triangleq\{E'\in P\suchthat e\rightarrow
 E'\}$. This set $\rightarrow(e)_P$ is a partition of $\rightarrow(e)$, the set
 of the successors of $e$. Then, each time a state $e$ is involved at
 line~\ref{SplitUpdateData:16} this means that a block $E'$ in
 $\rightarrow(e)_P$ has been split in $C$ and in $D$. For a given state $e$
 there can be at most $|\rightarrow(e)|$ splits of $\rightarrow(e)_P$. Hence,
 the sum of the sizes of $Touched'$ for all executions of
 \ref{func:SplitUpdateData} is in $O(|\rightarrow|)$.  Furthermore, for one
 execution of this procedure, the time complexity of lines
 \ref{SplitUpdateData:18}--\ref{SplitUpdateData:19} is in
 $O(|P_{\mathrm{sim}}|)$.  Therefore, the overall time complexity of lines
 \ref{SplitUpdateData:17} to \ref{SplitUpdateData:19} is in
 $O(|P_{\mathrm{sim}}|.|\rightarrow|)$.

Lines \ref{SplitUpdateData:21} to \ref{SplitUpdateData:28} treat the case where
the split block is $E$. The block $E$ has been split in $C$ and $D$. One of them
contains $E.rep$. It is therefore not necessary to recompute the counters
associated with it (although these counters have been possibly updated at lines
\ref{SplitUpdateData:12} to \ref{SplitUpdateData:20}). The variable $X$
represents the block, $C$ or $D$, which has not inherited $E.\Rep$. Therefore, we have to
initialize the counters associated with it. Note that lines
\ref{SplitUpdateData:21} to \ref{SplitUpdateData:28} are executed at most once
for each block. Therefore, apart from the nested loops at lines
\ref{SplitUpdateData:25}--\ref{SplitUpdateData:27}, the overall time complexity
of them is in $O(|P_{\mathrm{sim}}|.|\rightarrow|)$.  For the nested loops, we
have to observe that the size of $Touched$ is less than the number of outgoing
transitions from $X.\Rep$ and those transitions are considered only once during
the execution of the algorithm. Therefore the overall time complexity of the
nested loops is also in $O(|P_{\mathrm{sim}}|.|\rightarrow|)$.

All of this implies that the  overall time complexity of Procedure
\ref{func:SplitUpdateData} is in $O(|P_{\mathrm{sim}}|.|\rightarrow|)$.

\subsubsection{Function \ref{func:Init}}
\label{sec:Init}

\begin{function}[!t]
  \caption{Init($Q,\rightarrow,P_{\mathrm{init}},R_{\mathrm{init}}, RefinerNodes$)}
  \label{func:Init}
  $P := \Copy(P_{\mathrm{init}})$ \nllabel{Init:1}\;
  \ForEach{$E\in P$}
  {
    $E.\Count := 0$ ; $E.\Rep := \ChooseState{E}$ \;
    $E.\Node := \{q\in Q\suchthat q.\Block = E\}$ \;
    $E.\Rel := \{E' \in P \suchthat (E,E') \in R_{\mathrm{init}}\}$ \nllabel{Init:5}\;
  }

  \BlankLine 
  \tcp{\small Initialization to take into account
    Proposition~\ref{prop:InitRefine}.} 
 
  $Marked := \emptyset$ ; $Touched := \emptyset$ \nllabel{Init:6}\;
  \lForEach{$q\rightarrow q'$}
  {
    $Marked := Marked \cup\{q\}$ }
                        
  \texttt{\ref{func:Split}$(P,Marked)$}\nllabel{Init:8} \;
  \ForEach{$E\in P\suchthat E.\Rep\in Marked$}
    {$Touched := Touched \cup\{E\}$ ;}

  \ForEach{$C\in Touched$}
  {
    \lForEach{$D\not\in Touched$}
    { $C.\Rel := C.\Rel \mathrel\setminus\{D\}$ }
  }
  $Marked := \emptyset$ ; $Touched := \emptyset$\nllabel{Init:13} \;
 
  \BlankLine
  \tcp{\small Initialization of $RefinerNodes$ and the $\NotRel$ 's.}
  \ForEach{$C\in P$\nllabel{Init:14}}
  {
    $C.\Node.\NotRel' :=\emptyset$ \;
    $C.\Node.\NotRel := \{D.node \suchthat D \not\in C.\Rel\}$ \;
    \If{$C.\Node.\NotRel \neq\emptyset$}
    {$RefinerNodes:=RefinerNodes\cup\{C.\Node\}$ \;}
  }
    
  \BlankLine
  \tcp{\small Initialization of  the \RelCountAlgo's.}
  \lForEach{$E,B'\in P$}  {$B'.\RelCountAlgo{E} := 0$ \nllabel{Init:19}}
  $PreE' := \emptyset$ \;
  \ForEach{$E'\in P$}
  {
    \ForEach{$e\in\rightarrow^{-1}(E'.\Node)\suchthat e=e.\Block.\Rep$\nllabel{Init:22}}
    {$PreE' := PreE'\cup \{e.\Block\}$ ;}
    
    \ForEach{$E\in PreE', B'\in P$\nllabel{Init:24}}
    {
      $B'.\RelCountAlgo{E}++$ ;
    }
    $PreE' := \emptyset$\nllabel{Init:26} \;
  }
  
  \Return{$(P)$} \;
\end{function}

This function initializes the data structures and transforms the initial preorder
such that we start from a preorder stable with the total relation
$\mathscr{R}_0=Q\times Q$.  The first lines require no special comment except
that the $\Rel$ array for each block $E$ is initialized according to
\eqref{eq:defRRel}. The time complexity of these lines
\ref{Init:1}--\ref{Init:5} is in  $O(|P_{\mathrm{sim}}|^2)$.

Lines \ref{Init:6}--\ref{Init:13} transform the initial preorder according to
Proposition~\ref{prop:InitRefine}. Note that the call of Function
\ref{func:Split} at line~\ref{Init:8} has the side effect to transform the
counters before their initialization. This is not a problem since this does not
change the overall time complexity and the real initialization of the counters
is done after, at lines \ref{Init:19}--\ref{Init:26}. Just after
line~\ref{Init:13}, with Proposition~\ref{prop:InitRefine}, the invariants
\eqref{eq:RelNotRelStability} is true for $i=1$ and each simulation included in
the preorder represented by the partition-relation pair
$(P_{\mathrm{init}},R_{\mathrm{init}})$ is included in $\mathscr{R}_1$.  Apart
from the inner call of Procedure \ref{func:SplitUpdateData}, whose we know the
overall time complexity, $O(|P_{\mathrm{sim}}|.|\rightarrow|)$, the time
complexity of these lines is in $O(|\rightarrow| + |P_{\mathrm{sim}}|^2)$.

The loop starting at line \ref{Init:14} initializes the \NotRel's according to
\eqref{eq:defNotRelAlgo} such that \eqref{eq:RiDecomposition} is also true for
$i=1$.  The set $RefinerNodes$ is also initialized according to
\eqref{eq:RefinerNodes} for $i=1$. The time complexity of this loop is in
$O(|P_{\mathrm{sim}}|^2)$.

Lines \ref{Init:19}--\ref{Init:26} initialize the counters such that for all
$E,B'\in P$ we have:
\begin{displaymath}
B'.\RelCountAlgo{E} = |\{E'\in P\suchthat E.\Rep\rightarrow
E'.\Node\}|
\end{displaymath}
The invariant \eqref{eq:SplitUpdateData} is thus true for $i=1$ since
$B'.\Node\times E'.\Node \subseteq \mathscr{R}_0$ is always true, with $\mathscr{R}_0=Q\times Q$.
The time complexity of the loop at line \ref{Init:19} is in
$O(|P_{\mathrm{sim}}|^2)$. The time complexity of the loop at line \ref{Init:22} is in
$O(|P_{\mathrm{sim}}|.|\rightarrow|)$. An iteration of the loop at line \ref{Init:24}
corresponds to a meta-transition $E.\Node \rightarrow E'.\Node$ and a block
$B'\in P$. Its time complexity is therefore in $O(|P_{\mathrm{sim}}|.|\rightarrow|)$.

\subsubsection{Procedure \ref{func:Split1}}
\label{sec:Split1}

\begin{procedure}[!t]
  % \SetAlgoLined
  \SetAlgoVlined
  \caption{Split1($P, RefinerNodes$)}
  \label{func:Split1}

  $Marked := \emptyset$ ; atLeastOneSplit := \False \ \; 
  \ForEach{$B\in RefinerNodes $}
  {
    $B'=\ChooseBlock(B)$ \nllabel{Split1:3}\;
    \ForEach{$e\in\rightarrow^{-1}(B)$ \nllabel{Split1:4}}
    {
      \If {$B'.\RelCountAlgo(e.\Block)=0$}
      {
        $atLeastOneSplit := \True$ \;
      }
    }
    \If {$atLeastOneSplit = \True$}
    {
      \ForEach{$q\rightarrow q' \suchthat q'.\Block \in B'.\Rel$\nllabel{Split1:8}}
        {
          $Marked := Marked\cup \{q\}$ \;
        }
      \texttt{\ref{func:Split}($P,Marked$)} \;
      $Marked := \emptyset$ ; atLeastOneSplit := \False  \nllabel{Split1:11}\;
    }      
  }
\end{procedure}

The purpose of this procedure is to apply Lemma~\ref{lem:split1} to split the
current partition $P$ until there is no $(\mathscr{P},\mathscr{R})$-splitter
transition of type 1, with $\mathscr{P}=\mathscr{P}_P$ and $\mathscr{R}$ defined
by \eqref{eq:defRRel}. This is done such that the coarsest $\mathscr{R}$-block-stable
equivalence relation included in $\mathscr{P}_P$ before the execution of
\ref{func:Split1} is still included in $\mathscr{P}_P$ after its
execution. Proposition~\ref{prop:notRel} guarantees that all splitter
transitions of type 1 have been treated since we assume
\eqref{eq:RefinerNodes}. Note that at line \ref{Split1:3}, $B$ is a node, a
set of states, which corresponds to a block of the relation $\mathscr{R}$. As
explained in section \ref{sec:data-structures}, the counters are defined between
two blocks of the current partition $P$. We thus need to choose one of this block
included in $B$ to represent $B$ since we have Proposition~\ref{prop:RelCount}.

A transition $e\rightarrow e'$ with $e'\in B$ at line \ref{Split1:4} is
considered only if there is a block in $\cup B.\NotRel$. From Lemma~\ref{lem:InOnlyOne}
this can happen only $|P_{\mathrm{sim}}|$ times. Therefore, the overall time complexity
of the loop at line \ref{Split1:4} is in $O(|P_{\mathrm{sim}}|.|\rightarrow|)$.  From
Lemma~\ref{lem:split1}, lines \ref{Split1:8}--\ref{Split1:11} are executed only
when at least one block is split. Therefore, their overall time complexity is in
$O(|P_{\mathrm{sim}}|.|\rightarrow|)$.

\subsubsection{Procedure \ref{func:Split2}}
\label{sec:Split2}

\begin{procedure}[!t]
  \caption{Split2($P, RefinerNodes$)}
  \label{func:Split2}
    
  $PreB' := \emptyset$ ; $Touched := \emptyset$ ; $Touched' := \emptyset$ ;
  $Marked := \emptyset$ \;

  \ForEach{$B\in RefinerNodes $}
  {
    \ForEach{$B' \in P \suchthat B'.\Node\subseteq B$\nllabel{Split2:3}}
    {
      \ForEach{$e\in \rightarrow^{-1}(B'.\Node)$\nllabel{Split2:4}}
      {
        $E := e.\Block$ \;
        \If {$e=E.\Rep$}
        {
          $PreB' := PreB' \cup \{E\}$ \;          
        }
      }
      \ForEach{$E \in PreB'$\nllabel{Split2:8}}
      {
        $E.\Count$++ \;
        $Touched := Touched \cup \{E\}$ \;                
      }
      $PreB' := \emptyset$ \;
      
    }
    \BlankLine
    $B'=\ChooseBlock(B)$ \;
    \ForEach{$E\in Touched$\nllabel{Split2:13}}
    {
      \If{$B'.\RelCountAlgo(E) = E.\Count$}
      {
        \tcp{\small The transition $E.\Rep \rightarrow B$ is maximal}
        $Touched' := Touched' \cup \{E\}$ \;
      }
      $E.\Count \mathrel{:=} 0$  \;
    }
     
    \ForEach{$e\in \rightarrow^{-1}(B) \suchthat e.\Block \in Touched'$\nllabel{Split2:17}}
    {
      $Marked := Marked\cup \{e\}$ \;
    }
    
    \texttt{\ref{func:Split}($P,Marked$)} \nllabel{Split2:19}\;

      $Touched := \emptyset$ ;  $Touched' := \emptyset$ ; $Marked := \emptyset$ \;
  }
\end{procedure}

This procedure is applied after \ref{func:Split1}. We can therefore assume that
there is no more $(\mathscr{P},\mathscr{R})$-splitter transition of type 1, with
$\mathscr{P}$ and $\mathscr{R}$ defined in the analysis of Split1. The aim of
this procedure is to implement Lemma~\ref{lem:split2}. It works as follows. For
a given node $B$ which corresponds to a block of $\mathscr{R}$ such that
$B.\NotRel\neq\emptyset$ we scan the blocks $B'$ of $P$ which are included in
$B$. For each of those $B'$ we scan, loop at line \ref{Split2:4}, the incoming
transitions from representatives states of blocks $E$. For each of these blocks
we increment $E.\Count$, loop at line \ref{Split2:8}. Therefore, at the end of
the loop at line \ref{Split2:3}, for each block $E\in P$ such that $E\rightarrow
B$, we have $E.\Count = |\{[b]_{\mathscr{P}}\subseteq B\suchthat
E.\Rep\rightarrow b\}|$. This allows us to identify, loop at
line~\ref{Split2:13}, the blocks that may be split according to
Lemma~\ref{lem:split2}. The split is done thanks to
lines\ref{Split2:17}--\ref{Split2:19}.  For reasons similar to those for
Procedure \ref{func:Split1}, we have: the overall time complexity of Procedure
\ref{func:Split2} is in $O(|P_{\mathrm{sim}}|.{|\rightarrow|})$.

\subsubsection{Procedure \ref{func:Refine}}
\label{sec:Refine}

\begin{procedure}[!t]
  % \SetAlgoLined
  % \SetAlgoVlined  
  \caption{Refine($RefinerNodes$)}
  \label{func:Refine}
  
  $Remove := \emptyset$  ; $RefinerNodes':=\emptyset$ \;
  \ForEach{$B\in RefinerNodes $\nllabel{Refine:2}}
  {
    $B'=\ChooseBlock(B)$ \;
   \ForEach{$d\in \rightarrow^{-1}( \cup B.\NotRel)$\nllabel{Refine:4}}
   {
     $D := d.\Block$ \;
      \If {$B'.\RelCountAlgo{D} = 0$\nllabel{Refine:6}}
      {
        $Remove := Remove \cup \{D\}$ \;
      }
    }
    \ForEach{$c\in \rightarrow^{-1}(B)$\nllabel{Refine:8}}
    {
      $C := c.\Block$ \;
     \ForEach{$D\in Remove \suchthat D\in C.\Rel$\nllabel{Refine:10}}
      {
        $C.\Rel := C.\Rel \mathrel\setminus \{D\}$ \;
        $C.\Node.\NotRel' := C.\Node.\NotRel' \cup \{D.\Node\}$ \;
        $RefinerNodes' := RefinerNodes' \cup \{C.\Node\}$ \;
      }        
    }
    $Remove := \emptyset$ ; $B.\NotRel := \emptyset$ \;
  }

  $RefinerNodes:=\emptyset$ \nllabel{Refine:15}\;
  \ForEach{$B\in RefinerNodes'$}
  {
    $\Swap(B.\NotRel, B.\NotRel')$ \;
  }
  $\Swap(RefinerNodes, RefinerNodes')$ \nllabel{Refine:18}\;
\end{procedure}
Procedures \ref{func:Split1} and \ref{func:Split2} possibly change the current
partition $P$, but not $\mathscr{R}$, and preserve \eqref{eq:defNotRelAlgo},
\eqref{eq:RelNotRelStability}, \eqref{eq:SimIncludedRecc},
\eqref{eq:RiDecomposition}, \eqref{eq:RefinerNodes} and
\eqref{eq:SplitUpdateData_i}.  Thanks to Lemma~\ref{lem:split1} and
Lemma~\ref{lem:split2} all $\mathscr{R}$-block-stable equivalence relation
presents in $\mathscr{P}_P$ before the execution of \ref{func:Split1} and
\ref{func:Split2} are still presents after. Furthermore, with
Theorem~\ref{th:split1_2}, we know that $\mathscr{P}_P$, after the execution of
\ref{func:Split2}, is $\mathscr{R}$-block-stable. From all of this, before
the execution of \ref{func:Refine},
$\mathscr{P}_P$ is the coarsest $\mathscr{R}$-block-stable equivalence relation. The conditions are thus met to apply
Theorem~\ref{th:refinement} to do a refine step of the algorithm.  Note that,
thanks to \eqref{eq:SplitUpdateData_i} and Proposition~\ref{prop:RelCount}, we
have the equivalence: $d\not\in \rightarrow^{-1}\mathrel\circ\mathscr{R}(B)
\Leftrightarrow B'.\RelCountAlgo{D} = 0$ at line \ref{Refine:6}.  At the end of
the while loop at line \ref{Refine:2}, the relation $\mathscr{R}$ has been
refined by $\mathit{NotRel}'$. In lines \ref{Refine:15}--\ref{Refine:18} $\mathit{NotRel}$ is set
to $\mathit{NotRel}'$ and $RefinerNodes$ is set to $RefinerNodes'$ to prepare the next
iteration of the while loop in Procedure~\ref{func:Sim}. By construction, see
the loop starting at line \ref{Refine:10}, \eqref{eq:RiDecomposition} and
\eqref{eq:RefinerNodes} are set for the next iteration of the while loop in
\ref{func:Sim}. In a similar way, the \RelCount's are not modified by
\ref{func:Refine}, \eqref{eq:SplitUpdateData_i} is thus still true and
\eqref{eq:SplitUpdateData} will be true for the next iteration of the while loop
in \ref{func:Sim}. Thanks to Theorem~\ref{th:refinement},
\eqref{eq:RelNotRelStability} and \eqref{eq:SimIncludedRecc} are also preserved.

From Lemma~\ref{lem:InOnlyOne}, a node $B$ and a transition $d\rightarrow d'$
with $d'\in \cup B.\NotRel$ are considered at most once during the execution of
the algorithm. Therefore, the overall time complexity of the loop at line
\ref{Refine:4} is in $O(|P_{\mathrm{sim}}|.|\rightarrow|)$.  Let us now consider
a block $D$ in $Remove$ and a transition $c\rightarrow b$ with $b\in B$. By
contradiction, let us suppose this pair $(D,c\rightarrow b)$ can happen twice,
at iteration $i$ and at iteration $j$ of the while loop of Function
\ref{func:Sim}, with $i<j$. From \eqref{eq:RiDecomposition} and line
\ref{Refine:4} we have for $k=i$ and $k=j$: $ D\rightarrow \mathscr{R}_{k-1}(b)$
and $D\not\rightarrow \mathscr{R}_{k}(b)$. But this is not possible since
$(\mathscr{R}_i)_{i\geq 0}$ is a strictly decreasing sequence of relations. This
means that the overall time complexity of the loop at line \ref{Refine:8} is
also in $O(|P_{\mathrm{sim}}|.|\rightarrow|)$. The other lines have a lower
overall time complexity. From all of this, the overall time complexity of this
procedure is in $O(|P_{\mathrm{sim}}|.|\rightarrow|)$.

\subsubsection{Time Complexity of the Algorithm}
From the analysis of the functions and procedures of the algorithm, we derive
the following theorem.
\begin{theorem}
  The time complexity of the presented simulation algorithm is in
  $O(|P_{\mathrm{sim}}|.|\rightarrow|)$. 
\end{theorem}

%\newpage

\section{Improvements and Future Work}
\label{sec:future}
% Pour eviter un mauvais decoupage des expressions mathematiques
%\binoppenalty=\maxdimen
%\relpenalty=\maxdimen
With the new notions of maximal transition, stable preorder, block-stable
equivalence relation and representative state, we have introduced new
foundations that we will use to design some efficient simulation
algorithms. This formalism has been illustrated with the presentation of the
most efficient in memory of the fastest simulation algorithms of the moment.

It is possible to increase in practice the time efficiency of procedures
\ref{func:SplitUpdateData}, \ref{func:Split1} and \ref{func:Split2} if we allow the use of both
$\rightarrow$ and $\rightarrow^{-1}$ (or if we calculate one from the other,
which requires an additional bit space in $O({|\rightarrow|.}\log(|Q|))$ to store
the result). Procedures \ref{func:SplitUpdateData} and \ref{func:SimUpdateData}
can be further improved in practice, but this changes $O(|P_{\mathrm{sim}}|^2.\log(|P_{\mathrm{sim}}|))$ to
$O(|P_{\mathrm{sim}}|^2.\log(|Q|))$, which is not really noticeable in practice,
in the bit space complexity, if we count states 
instead of blocks in \eqref{eq:RelCount}, \eqref{eq:SplitUpdateData} and
\eqref{eq:SplitUpdateData_i}.

Simulation algorithms are generally extended for labeled transition systems
(LTS) by embedding them in normal transition systems. This is what is proposed
in \cite{RT10,GPP03} and \cite{GPP15} for example. By doing this, the size of
the alphabet is introduced in both time and space complexities. Even in
\cite{ABHKV08a}, where a more specific algorithm is proposed for LTS, the size of
the alphabet still matter. In \cite{Cec13a} we proposed three extensions of
\cite{RT10} for LTS with significant reduction of the incidence of the size of
the alphabet. We will therefore propose the same extensions but from the
foundations given in the present paper. Note that the bit space complexity of
the algorithm presented here is in fact in
$\Theta(|P_{\mathrm{sim}}|^2.\log(|P_{\mathrm{sim}}|) + |Q|.\log(|Q|))$. It is
therefore interesting to propose other algorithms with better compromises between time and space complexities. We will therefore compare in practice
different propositions.
% There
%  We will also implements these extensions. Note that our implementation
% (not yet published) of the fast version in \cite{Cec13a} is already competitive
% with the algorithm in \cite{RT10} and the \emph{nice} version in \cite{Cec13a} is also
% already competitive with the algorithm in \cite{GP08}.

Then, we will have all the prerequisites to address a more challenging problem
that we open here: the existence of a simulation algorithm with a time
complexity in $O(|\rightarrow|.\log(|Q|) +
|P_{\mathrm{sim}}|.|{\rightarrow_{\mathrm{sim}}}|)$, with
$\rightarrow_{\mathrm{sim}}$ the relation over $P_{\mathrm{sim}}$ induced by
$\rightarrow$, and a bit space complexity in
$O(|P_{\mathrm{sim}}|^2.\log(|P_{\mathrm{sim}}|)
+{|\rightarrow|.}\log(|Q|))$.
What is surprising is that the biggest challenge
is not the part in $O(|P_{\mathrm{sim}}|.|{\rightarrow_{\mathrm{sim}}}|)$\Olics{\newpage
\noindent} but
the part in $O(|\rightarrow|.\log(|Q|))$ in the time complexity.
Such an
algorithm will lead to an even greater improvement from the algorithm of the
present paper than that of the passage from HHK to RT since there are, in
general, many more transitions than states in a transition system.

%\vspace{1cm}
%\bline
\section*{Acknowledgements}

I thank the anonymous reviewers\OarXiv{ of the paper to LICS'17}. Their
questions and suggestions have helped to improve the presentation of the
paper. I am also grateful to Philippe Canalda for his helpful advice.

% Dans les complexite en espace, preciser qu'on est en fait dans :  $\Theta(|P_{\mathrm{sim}}|^2.\log(|P_{\mathrm{sim}}|) +|Q|.\log(|Q|)$,

% De \eqref{def:relStability},   \eqref{th:split1_2} on peut montrer qu'a la fin,
% $P_{\mathscr{R}}$ est une bisimulation. Par contre, il faut encore montrer que
% c'est la plus grande. 

%\newpage
%\bibliography{simulation}

\begin{thebibliography}{ACH{\etalchar{+}}10}

\bibitem[ABH{\etalchar{+}}08]{ABHKV08a}
Parosh~Aziz Abdulla, Ahmed Bouajjani, Luk{\'a}s Hol\'{\i}k, Lisa Kaati, and
  Tom{\'a}s Vojnar.
\newblock Computing simulations over tree automata.
\newblock In C.~R. Ramakrishnan and Jakob Rehof, editors, {\em TACAS}, volume
  4963 of {\em Lecture Notes in Computer Science}, pages 93--108. Springer,
  2008.

\bibitem[ACH{\etalchar{+}}10]{ACHMV10}
Parosh~Aziz Abdulla, Yu-Fang Chen, Luk{\'a}s Hol\'{\i}k, Richard Mayr, and
  Tom{\'a}s Vojnar.
\newblock When simulation meets antichains.
\newblock In Javier Esparza and Rupak Majumdar, editors, {\em TACAS}, volume
  6015 of {\em Lecture Notes in Computer Science}, pages 158--174. Springer,
  2010.

\bibitem[BG03]{BG03}
Doron Bustan and Orna Grumberg.
\newblock Simulation-based minimization.
\newblock {\em ACM Trans. Comput. Logic}, 4(2):181--206, 2003.

\bibitem[C{\'{e}}c13]{Cec13a}
G{\'{e}}rard C{\'{e}}c{\'{e}}.
\newblock Three simulation algorithms for labelled transition systems.
\newblock {\em CoRR}, http://arxiv.org/abs/1301.1638, 2013.

\bibitem[CG11]{CG11}
G{\'e}rard C{\'e}c{\'e} and Alain Giorgetti.
\newblock Simulations over two-dimensional on-line tessellation automata.
\newblock In Giancarlo Mauri and Alberto Leporati, editors, {\em Developments
  in Language Theory}, volume 6795 of {\em Lecture Notes in Computer Science},
  pages 141--152. Springer, 2011.

\bibitem[CRT11]{CRT11}
Silvia Crafa, Francesco Ranzato, and Francesco Tapparo.
\newblock Saving space in a time efficient simulation algorithm.
\newblock {\em Fundam. Inform.}, 108(1-2):23--42, 2011.

\bibitem[GL94]{GL94}
Orna Grumberg and David~E. Long.
\newblock Model checking and modular verification.
\newblock {\em ACM Trans. Program. Lang. Syst.}, 16(3):843--871, 1994.

\bibitem[GPP03]{GPP03}
Raffaella Gentilini, Carla Piazza, and Alberto Policriti.
\newblock From bisimulation to simulation: Coarsest partition problems.
\newblock {\em J. Autom. Reasoning}, 31(1):73--103, 2003.

\bibitem[GPP15]{GPP15}
Raffaella Gentilini, Carla Piazza, and Alberto Policriti.
\newblock Rank and simulation: the well-founded case.
\newblock {\em J. Log. Comput.}, 25(6):1331--1349, 2015.

\bibitem[HHK95]{HHK95}
Monika~Rauch Henzinger, Thomas~A. Henzinger, and Peter~W. Kopke.
\newblock Computing simulations on finite and infinite graphs.
\newblock In {\em FOCS}, pages 453--462. IEEE Computer Society, 1995.

\bibitem[Mil71]{Milner71}
Robin Milner.
\newblock An algebraic definition of simulation between programs.
\newblock In {\em IJCAI}, pages 481--489, 1971.

\bibitem[PT87]{PT87}
Robert Paige and Robert~Endre Tarjan.
\newblock Three partition refinement algorithms.
\newblock {\em SIAM J. Comput.}, 16(6):973--989, 1987.

\bibitem[Ran14]{Ran14}
Francesco Ranzato.
\newblock An efficient simulation algorithm on kripke structures.
\newblock {\em Acta Inf.}, 51(2):107--125, 2014.

\bibitem[RT07]{RT07}
Francesco Ranzato and Francesco Tapparo.
\newblock A new efficient simulation equivalence algorithm.
\newblock In {\em 22nd {IEEE} Symposium on Logic in Computer Science {(LICS}
  2007), 10-12 July 2007, Wroclaw, Poland, Proceedings}, pages 171--180. {IEEE}
  Computer Society, 2007.

\bibitem[RT10]{RT10}
Francesco Ranzato and Francesco Tapparo.
\newblock An efficient simulation algorithm based on abstract interpretation.
\newblock {\em Inf. Comput.}, 208(1):1--22, 2010.

\bibitem[vGP08]{GP08}
Rob~J. van Glabbeek and Bas Ploeger.
\newblock Correcting a space-efficient simulation algorithm.
\newblock In Aarti Gupta and Sharad Malik, editors, {\em CAV}, volume 5123 of
  {\em Lecture Notes in Computer Science}, pages 517--529. Springer, 2008.

\bibitem[VL08]{VL08}
Antti Valmari and Petri Lehtinen.
\newblock Efficient minimization of dfas with partial transition.
\newblock In Susanne Albers and Pascal Weil, editors, {\em STACS}, volume~1 of
  {\em LIPIcs}, pages 645--656. Schloss Dagstuhl - Leibniz-Zentrum fuer
  Informatik, Germany, 2008.

\end{thebibliography}
\newcommand{\etalchar}[1]{$^{#1}$}

% \newpage\mbox{}
% %\newpage\mbox{}
% %\appendix\Olics{[Proofs]}
% \Olics{\appendix[Proofs]}
% \OarXiv{\appendix
% \section{Proofs not given in the main text}
% \label{sec:proofs-not-given}
% }

%%%  ICI

\end{document}